\documentclass[11pt,a4paper]{article}

\usepackage{a4wide}
\usepackage[nonamelimits]{amsmath}
\usepackage{amsthm,amssymb,hyperref}
\usepackage{filecontents}
\usepackage{comment}
\usepackage[shortlabels]{enumitem}
\usepackage{float}
\usepackage{mathtools}
\usepackage{thm-restate}
 
\hypersetup{colorlinks=true,citecolor=blue, linkcolor=blue, urlcolor=blue}

\makeatletter
\def\prob#1#2#3{\goodbreak\begin{list}{}{\labelwidth\z@ \itemindent-\leftmargin
                        \itemsep\z@  \topsep6\p@\@plus6\p@
                        \let\makelabel\descriptionlabel}\cG
                \item[\it Name]#1
               \item[\it Instance]                #2
                \item[\it Output]#3
                \end{list}}
\makeatother 

\newtheorem{theorem}{Theorem}
\newtheorem{lemma}[theorem]{Lemma}
\newtheorem{corollary}[theorem]{Corollary}
\theoremstyle{definition}
\newtheorem{definition}[theorem]{Definition}
\newtheorem{observation}[theorem]{Observation}

\newtheorem{example}{Example}

\newcommand{\Tmix}{T_{\mathrm{mix}}}

\newcommand{\eps}{\epsilon}

\def\E{\mathbb{E}}

\def\Z{\mathbb{Z}}

\def\eps{\varepsilon}

\def\gam{\gamma}

\def\cC{\mathcal{C}}
\def\cP{\mathcal{P}}

\def\cA{\mathcal{A}}

\def\cG{\mathcal {G}}

\def\1{\mathbf{1}}

\def\lam {\lambda}

\def\PD{\mathcal{M}'}
\def\rGD{\mathcal{M}}

\allowdisplaybreaks
 
\title {Fast algorithms at low temperatures via Markov chains\thanks{These results were announced in preliminary form (without proofs) as a brief abstract in the proceedings of APPROX/RANDOM 2019}}
 
\author{
Zongchen Chen\thanks{School of Computer Science, Georgia Institute of Technology.  
Research supported in part by NSF grants CCF-1617306 and CCF-1563838.}
 \and 
Andreas Galanis\thanks{The research leading to these results has received funding from the European Research Council under
  the European Union's Seventh Framework Programme (FP7/2007-2013) ERC grant agreement no.\ 334828. The paper
  reflects only the authors' views and not the views of the ERC or the European Commission.
  The European Union is not liable for any use that may be made of the information contained therein.
  Authors' address: Department of Computer Science, University of Oxford, Wolfson Building, Parks Road, Oxford, OX1~3QD, UK.}
  \and
  Leslie Ann Goldberg$^\dag$
  \and Will Perkins\thanks{Department of Mathematics, Statistics, and Computer Science, University of Illinois at Chicago.  Supported in part by NSF grants DMS-1847451 and CCF-1934915. Part of this work was done while WP was visiting the Simons Institute for the Theory of Computing.}
  \and James Stewart$^\dag$ \and 
  Eric Vigoda$^*$}

\date{April 13, 2021}

\def\calI{\mathcal{I}}
 
\newcommand{\Vpart}[1]{V^{#1}}
  
 \newcommand{\ZhcG}{Z_{G, \lambda}}
 \newcommand{\muhc}{\mu_\lambda}
\newcommand{\muhcG}{\mu_{G, \lambda}}
 \newcommand{\Cpart}[1]{\cC^{#1}}
 \newcommand{\wpart}[1]{w^{#1}}
 \newcommand{\opposite}{1-}
 \newcommand{\ceil}[1]{\left\lceil #1 \right\rceil}

\let \epsilon=\varepsilon

\def\calM{\mathcal{M}}
\def\PSC{\tau}
 
 \def\E{\mathbb{E}}

 \def\Z{\mathbb{Z}}
 
 \def\eps{\varepsilon}
 
 \def\gam{\gamma}

 \def\cC{\mathcal {C}}
 \def\cG{\mathcal {G}}
 
\newcommand{\cCtrunc}[1]{\mathcal{C}_{#1}}

\def\e{e}

\def\cSexact{\mathcal{S}_{\mathrm{exact}}}
 
 \def\e{e}
 
 \def\N{\mathbb{N}}
 
 \def\lam {\lambda}

\begin{document}

\maketitle

\begin{abstract} 
Efficient algorithms for approximate counting and sampling in spin systems typically apply in the so-called high-temperature regime, where the interaction between neighboring spins is ``weak''.   Instead, recent work of Jenssen, Keevash and Perkins yields polynomial-time algorithms in the low-temperature regime on bounded-degree (bipartite) expander graphs using polymer models and the cluster expansion.

In order to speed up these algorithms (so the exponent in the run time does not depend on the degree bound) we present a Markov chain for polymer models and show that it is rapidly mixing under exponential decay of polymer weights. This yields, for example, an $O(n \log n)$-time sampling algorithm for the low-temperature ferromagnetic Potts model on bounded-degree expander graphs. Combining our results for the hard-core and Potts models with Markov chain comparison tools, we obtain polynomial mixing time for Glauber dynamics restricted to appropriate portions of the state space.
\end{abstract}

\section{Introduction}
\label{secIntro}

The hard-core model from statistical physics is defined on the set of independent sets of a graph $G$, where the
independent sets are weighted by a fugacity $\lambda>0$.  The associated Gibbs distribution $\mu_{G,\lambda}$ is 
defined as follows, for an independent set $I$:
\begin{align}
\muhcG(I) &= \frac{\lam^{|I|}}{ \ZhcG} 
\end{align}
where 
$\ZhcG = \sum_{I\in \calI(G)} \lam ^{|I|}$ is the hard-core partition function (also called the independence polynomial), $\calI(G)$ is the set of independent sets of~$G$,   and $\lam>0$ is the \textit{fugacity}.

In applications, there are two important computational tasks associated to a spin model such as the hard-core model.  Given an error parameter 
$\eps\in (0,1)$, 
an $\eps$-approximate counting algorithm  outputs a number $\hat Z$ so that $ e^{-\eps} \ZhcG  \le \hat Z \le e^{\eps} \ZhcG$, and an $\eps$-approximate sampling algorithm outputs a random sample $I$ with distribution $\hat \mu$ so that the total variation distance
satisfies $\| \muhc - \hat \mu \|_{TV} <\eps$.

While classical statistical physics is most interested in studying the hard-core model on the integer lattice  $\Z^d$, the perspective of computer science is to consider wider families of graphs, such as the set of all graphs, all graphs of maximum degree $\Delta$, or all bipartite graphs of maximum degree $\Delta$.

Almost all proven efficient algorithms for approximate counting and sampling from the hard-core model work for low fugacities (the weak interaction regime, akin to the low temperature regime of the Potts model).  In the high temperature regime there are at least three distinct algorithmic approaches to approximate counting and sampling:  Markov chains, correlation decay, and polynomial interpolation.  One striking advantage of the Markov chain approach is that the algorithms are much faster and simpler than the algorithms from the other approaches.  In particular, it is common for a Markov chain sampling algorithm to run in time $O(n \log n)$, e.g., see~\cite{DGIndSet,EHSVY}, while typical running times for algorithms based on correlation decay~\cite{Weitz,LiuLu} and polynomial interpolation~\cite{Barvinokbook} are $n^{O(\log \Delta)}$ where $\Delta$ is the maximum degree of the graph.  

In general there are no known efficient algorithms at low temperatures (high fugacities), but recently efficient algorithms have been developed for some special classes of graphs including subsets of $\Z^d$~\cite{HPR}, random regular bipartite graphs, and bipartite expander graphs in general~\cite{JKP,liao2019counting}.  What these bipartite graphs have in common is that for large enough~$\lam$, typical independent sets drawn from $\muhcG$ align closely with one side or the other of the bipartition (the two ground states).  This phenomenon is related to the phase transition phenomenon in infinite graphs, and implies the exponentially slow mixing time of local Markov chains~\cite{borgs1999torpid,GT,mossel2009hardness}.     The algorithms introduced in~\cite{HPR} exploit this phenomenon  by expressing the partition function $\ZhcG$ in terms of deviations from the two ground states, and then using a truncation of a convergent series expansion (the Taylor series or the cluster expansion) to approximate the log partition function.  In statistical physics this is called a \textit{perturbative} approach, and while in general it does not work in the largest possible range of parameter space, when it does work it gives a very detailed probabilistic understanding of the model~\cite{PS,borgs1989unified,dobrushin1996estimates}.

To apply the so-called perturbative approach at low temperatures, one rewrites the original spin model as a new model in which single spin interactions are replaced by the interaction of connected components representing deviations from a chosen ground state.  Such models are known in general as \textit{abstract polymer models} \cite{KP}, see Section~\ref{sec:absolmod} for the polymer models we consider here, and have long been used in statistical physics to understand phase transitions.  In this paper, we  show that  once a low-temperature spin model has been transformed into a polymer model, Markov chains once again become an effective algorithmic tool.  Using this approach we obtain nearly linear and quadratic time sampling algorithms for low temperature models on expander graphs in cases where  only $n^{O(\log \Delta)}$-time algorithms were previously known.

\subsection{Subset polymer models}\label{sec:absolmod}

Abstract polymer models, as defined by Koteck\'{y} and Preiss \cite{KP},  are an important tool in studying the equilibrium phases of statistical physics models on lattices, see, e.g.,~\cite{laanait1991interfaces,borgs1989unified} among many others.\footnote{See also the relevant notion of `animal models' by Dobrushin~\cite{dobrushin1996estimates}.} The paper~\cite{borgs2006absence} has a more detailed history of their use in statistical physics and combinatorics. Recently, polymer models have been used to develop efficient deterministic algorithms for sampling and approximating the partition functions of statistical physics models on lattices~\cite{HPR} and expander graphs~\cite{JKP,liao2019counting} at low temperatures, the regime in which Markov chains like the Glauber dynamics are known to mix slowly.

We will study the following class of abstract polymer models, known as subset polymer models (defined by Gruber and Kunz~\cite{gruber1971general}). We begin by describing the relevant polymers: start with a finite host graph $G$ and a set $[q] = \{0,\ldots,q-1\}$ of spins. For each vertex~$v$, 
there is a ground-state spin $g_v$.  A polymer $\gamma$ consists of a 
connected set 
of vertices together with an assignment $\sigma_{\gamma}$ 
of  spins from $\{0,\ldots,q-1\}\setminus g_v$ to each vertex $v \in \gamma$ (we abuse notation and use $\gamma$ to denote both the polymer and the associated set of vertices).  The size of a polymer, $|\gamma|$, is the number of vertices in $ \gamma$. 
The set of all polymers is $\cP(G)$.

A polymer model on $G$ consists of a set $\cC(G)\subseteq \cP(G)$ of `allowed' polymers, and a non-negative weight $w_\gamma$ for each polymer $\gamma \in \cC(G)$. We denote this model by $(\cC(G), w)$. Two polymers $\gamma$ and $\gamma'$ are called `compatible' (written $\gamma \sim \gamma'$) 
if their distance in the host graph is at least~$2$; otherwise they are `incompatible' (written $\gamma \nsim \gamma'$).  The state space of allowable configurations is
$\Omega = \{ \Gamma \subseteq \cC(G) \mid \forall \gamma ,\gamma' \in \Gamma, \gamma \sim \gamma'\}$.

The partition function of the polymer model is 
\begin{align*}
Z(G) &= \sum_{\Gamma \in \Omega} \prod_{\gamma \in \Gamma} w_{\gamma},
\end{align*}
where the empty set of polymers contributes $1$ to the partition function.  The Gibbs measure $\mu_G$ is the probability distribution on $\Omega$ given by
\begin{align*}
\mu_G(\Gamma) &= \frac{\prod_{\gamma \in \Gamma} w_{\gamma}}{Z(G)} \,.
\end{align*}

Note that the polymer model is in fact a hard-core model on the `incompatibility graph' of $\cC(G)$, where two polymers are joined by an edge if they are  incompatible, with non-uniform fugacities given by the weights $w_\gamma$.  The geometry inherited from the host graph $G$ and the sizes of the polymers adds additional structure to the model. 

\begin{example}
\label{egHClow}
One instance of a polymer model is the hard-core model itself: polymers are single vertices of the graph $G$, labeled with `1' (for occupied) against a ground state `0' (for unoccupied).  Each polymer (vertex) $v$ comes with the weight function $w_v = \lam$.  Then the set of allowable polymer configurations is exactly the set of independent sets of $G$, and so the polymer model partition function is exactly the partition function of the hard-core model on~$G$. 
\end{example}

\begin{example}
\label{egPotts}
A second instance of a polymer model is related to the ferromagnetic $q$-color Potts model on a graph $G$ (see Definition~\ref{defPotts} below).  Fix a color $g \in [q]$ to be the ground state color, and define polymers to be connected subgraphs of $G$ of size at most $M$, with vertices labeled by the remaining colors $[q] \setminus \{ g\}$.  A polymer $\gamma$ has weight function $w_{\gamma} = e^{-\beta B(\gamma) }$ where $B(\gamma)$ is the number of bichromatic edges in $ \gamma$ plus the size of the edge boundary of $ \gamma$ in $G$.  A configuration of compatible polymers maps to a unique Potts configuration $\sigma$ in which all connected components of non-$g$-colored vertices have size at most $M$, and the weight of $\sigma$ in the Potts model is exactly the product of the weight functions of the polymers.  The polymer model partition function $Z(G)$, with an appropriate choice of $M$,  represents the contribution to the Potts model partition function of colorings where color $g$ `dominates', see also Section~\ref{sec:firstcomparison} for more details.   
\end{example}

As with the hard-core model, there are two main computational problems associated to a polymer model: approximate sampling from $\mu_G$ and approximate counting of $Z(G)$. We will approach them both via Markov chain algorithms.  In general we will be interested in families of polymer models defined on classes of graphs.  We denote such a family $(\cC(\cdot), w,  \cG)$, where for each graph $G \in  \cG$, $(\cC(G), w)$ is a polymer model.  We will always use $n$ to denote the number of vertices of a graph $G$.  

 We consider two conditions on the weight functions $w_{\gamma}$ and give their algorithmic consequences.

\begin{definition}
A polymer model $(\cC(\cdot), w, \cG)$ satisfies the \textit{polymer mixing condition} if there exists $\theta \in (0,1)$ such that
\begin{align}
\label{eqMixCondition}
\sum_{\gamma' \nsim \gamma} |\gamma'| w_{\gamma'} \le \theta |\gamma| 
\end{align}
 for all $G \in \cG$ and all $\gamma \in \cC(G)$.
\end{definition}

We postpone the formal definition of mixing time to Section~\ref{secpolyMCsec} and state our first main result here. 

\begin{restatable}{theorem}{thmPolyMix}
\label{thmPolyMix}
Suppose that a polymer model $(\cC(\cdot), w, \cG)$ satisfies the polymer mixing condition~\eqref{eqMixCondition}. Then for each $G \in \cG$ there is a Markov chain making single polymer updates with stationary distribution $\mu_G$ and mixing time 
 $\Tmix(\eps)= O(n \log (n/\epsilon))$.  
\end{restatable}

Theorem~\ref{thmPolyMix}  on its own does not guarantee an efficient 
algorithm for sampling from $\mu_G$ because 
the Markov chain only yields an efficient sampling algorithm
if we can implement each step efficiently.  We will show that under a stronger condition we can do this. 

\begin{definition}
A polymer model $(\cC(\cdot), w, \cG)$  is said to be \emph{computationally feasible} if, for each $G \in \cG$ and each $\gamma \in \cP(G)$, we can determine, in time polynomial in $|\gamma|$, whether $\gamma \in \cC(G)$, and compute $w_\gamma$ if it is.
\end{definition}

\begin{definition}
A polymer model $(\cC(\cdot), w, \cG)$ 
with $q$ spins on a class $\cG$ of graphs of maximum degree $\Delta$ satisfies the \textit{polymer sampling condition} with constant  $\PSC \ge 5 + 3 \log ((q-1)\Delta)$ if 
\begin{align}
\label{eqSampleCondition}
w_{\gamma} &\le e^{- \PSC | \gamma |}.
\end{align}
for all $G \in \cG$ and all $\gamma \in \cC(G)$.
\end{definition}

We have the following theorem.
 
\begin{restatable}{theorem}{thmPolySample}
	\label{thmPolySample}
	 
If a computationally feasible polymer model $(\cC(\cdot), w, \cG)$ satisfies the polymer sampling condition~\eqref{eqSampleCondition} then for all $G \in \cG$ there is an $\eps$-approximate sampling algorithm for $\mu_G$ with running time $O(n \log (n/\eps) \log(1/\epsilon))$.  	
\end{restatable}

Note that the polymer sampling condition required by Theorem~\ref{thmPolySample} is in general more demanding than the zero-freeness required by cluster-expansion algorithms, but, as Theorem~\ref{thmPolySample} demonstrates, it leads to faster algorithms. Finally, we can use the sampling algorithm and simulated annealing to give a fully polynomial time randomized approximation scheme (FPRAS) for computing the partition function of polymer models.

\begin{restatable}{theorem}{thmPolyCount}
\label{thmPolyCount}
If a computationally feasible polymer model $(\cC(\cdot), w, \cG)$ satisfies the polymer sampling condition~\eqref{eqSampleCondition} then for all $G \in \cG$ there is a randomized $\eps$-approximate counting algorithm for $Z(G)$ with running time $O((n/\eps)^2 \log^3 (n/\eps))$ and success probability at least $3/4$.
\end{restatable}

Fern{\'a}ndez, Ferrari, and Garcia~\cite{fernandez2001loss} introduced a condition very similar to the polymer mixing condition in the setting of polymer models on $\Z^d$.  Their objective was to derive probabilistic properties of polymer models directly, without going through the combinatorics and complex analysis inherent in the cluster expansion for the log partition function.  They introduced a continuous time stochastic process whose stationary distribution was the infinite volume Gibbs measure of their polymer model and their version of condition~\eqref{eqMixCondition} implied an exponentially fast rate of convergence of this process.  They remarked that such an approach had the potential to be an efficient computational tool.

Here we take an algorithmic point of view, and use the polymer mixing and sampling conditions to show that a simple discrete time Markov chain mixes rapidly and can be used to design efficient sampling and approximation algorithms.  Our approach differs from that of~\cite{fernandez2001loss} in that while they are interested primarily in the probabilistic properties of spin models on $\Z^d$, we are interested in algorithmic problems involving spin models on general families of graphs.  Our setting of discrete time processes on finite graphs is also more suitable to studying algorithmic questions.  Our work confirms the central point of~\cite{fernandez2001loss}: that complex analysis and absolute convergence of the cluster expansion is not necessary to derive many important properties of a polymer model.

\subsection{Applications}
\label{subsec:applications}

We apply our results for subset polymer models to two specific examples: the ferromagnetic Potts model and 
the hard-core model on expander graphs.  To state these results we need some definitions.

\begin{definition}
Let $\alpha>0$. A graph $G$ is an $\alpha$-expander graph if for all $S \subset V(G)$ with $|S| \le |V(G)|/2$, we have $e(S, S^c) \ge 
\alpha |S|$, where $S^c = V(G)\setminus S$ and
$e(S,S^c)$ is the number of edges exiting the set $S$. 
\end{definition}

\begin{definition}
\label{defPotts}
The $q$-color ferromagnetic Potts model with parameter $\beta>0$ is a random assignment of $q$ colors to the vertices of a graph defined by
\begin{align*}
\mu_{G, \beta}(\sigma) &= \frac{e^{-\beta m(G,\sigma) }}{ Z_{G, \beta}  }
\end{align*}
where $m(G,\sigma)$ is the number of bichromatic edges of $G$ under the coloring $\sigma$ and $Z_{G, \beta} = \sum_{\sigma} e^{-\beta m(G,\sigma)}$ is the Potts model partition function. The parameter $\beta$ is known as the inverse temperature.
\end{definition}

Jenssen, Keevash, and Perkins~\cite{JKP} gave an FPTAS and  polynomial-time sampling algorithm for the Potts model on expander graphs, with an algorithm based on the cluster expansion and Barvinok's method of polynomial interpolation.  Under essentially the same conditions on the parameters we give a Markov chain based sampling algorithm with near linear running time.

\begin{restatable}{theorem}{thmPotts}	
	\label{thmPotts}
Suppose $q\geq 2$, $\Delta \geq 3$ are integers and $\alpha>0$ is a real. Then for $\beta \ge  \frac{5+3\log((q-1) \Delta)}{\alpha}$ and any $qe^{-n} \leq \eps < 1$,
	there is an $\eps$-approximate sampling algorithm for the $q$-state ferromagnetic Potts model with parameter $\beta$ on all
	$n$-vertex  $\alpha$-expander graphs of maximum degree $\Delta$ with running time $O(n \log(n/\eps) \log (1/\epsilon))$.  There is also an $\eps$-approximate counting algorithm with running time $O( (n/\epsilon)^2 \log^3 (n/\eps))$ and success probability at least 3/4. 
\end{restatable}
Note that,  if the desired error satisfies $\eps < qe^{-n}$, then we can simply compute the partition function by brute force in poly$(n/ \epsilon)$ time. This observation combined with the above result gives an FPRAS, but we can no longer guarantee a running time of $O( (n/\epsilon)^2 \log^3 (n/\eps))$ for exponentially small values of $\eps$. A similar point also applies to the algorithm that we give for the hard-core model.

\begin{definition}
\label{defbipartiteexpander}
Let $\alpha\in (0,1)$. A bipartite graph $G=(V,E)$ with bipartition $V=\Vpart{0}\cup\Vpart{1}$
 is a bipartite $\alpha$-expander  
if, for $i\in \{0,1\}$ and all $S\subseteq \Vpart{i}$ where 
$|S| \leq |\Vpart{i}|/2$,
 we have $N_G(S) \ge (1+\alpha) |S|$ where $N_G(S)$ denotes the set of vertices that are adjacent to some vertex in $S$.
\end{definition}

Again we give a fast Markov chain based algorithm for sampling from the hard-core model for essentially the same range of parameters for which an FPTAS is given in~\cite{JKP}. 

\begin{restatable}{theorem}{thmHC}
\label{thmHC}
 Suppose $\Delta\geq 3$ is an integer and $\alpha\in(0,1)$ is a real. Then for any $\lambda \ge (3\Delta)^{6/\alpha}$ and $4e^{-n} \le \eps < 1$, there is an $\eps$-approximate sampling algorithm for the hard-core model with parameter $\lambda$ on all $n$-vertex bipartite $\alpha$-expander graphs of maximum degree $\Delta$. There is also an $\eps$-approximate counting algorithm for the hard-core model  with success probability at least $1-\epsilon$. Both algorithms run in time $O((n/\eps)^2 \log^3 (n/\eps) \log(1/\epsilon))$.
\end{restatable}

The extra factor in the running time of the sampling algorithm for the hard-core model as compared to the Potts model is due to the fact that the hard-core model on a bipartite graph does not in general exhibit exact symmetry between the ground states, and so we must approximate the partition functions of the even and odd dominant independent sets to sample.

We can extend these algorithms to obtain fast sampling algorithms in most situations in which a counting problem can be put in the framework of subset polymer models.  For instance, we can use Theorems~\ref{thmPolySample} and~\ref{thmPolyCount} to improve the running times of 
the algorithms given by~\cite{JKP2,liao2019counting} for sampling and counting proper $q$-colorings in $\Delta$-regular bipartite graphs (for large~$\Delta$).
The two papers give slightly different polymer models for proper $q$-colorings on $\Delta$-regular bipartite graphs ---
see \cite[Section 5]{JKP2} and \cite[Section 5.2]{liao2019counting}.
Section~5.2 of~\cite{liao2019counting} shows that their polymer model is computationally feasible.
Section~5.1 of~\cite{JKP2} shows that their polymer model 
 satisfies the Koteck\'y-Preiss condition --- in fact, their proof establishes the polymer sampling condition~\eqref{eqSampleCondition}.
It is easy to see (by comparing the polymer weights) that the polymer model of~\cite{liao2019counting} therefore also satisfies the
polymer sampling condition.
Thus, we get the following corollary of Theorem~\ref{thmPolySample} and~\ref{thmPolyCount}.

 \begin{corollary} 
There is an absolute constant $C>0$ so that for all even $q \ge 3$, all $\Delta \ge C q^2 \log^2 q$ and all $\eps > e^{-n/(8q)}$, there is an $\eps$-approximate sampling algorithm to sample a uniformly random proper $q$-coloring from a random $\Delta$-regular bipartite graph running in time $O(n \log (n/\eps) \log(1/\epsilon))$. 
Furthermore, there is a randomized $\eps$-approximation algorithm for the number of proper $q$-colorings  with running time $O((n/\eps)^2 \log^3 (n/\eps))$ and success probability at least $3/4$.  For odd $q$, there are $\eps$-approximate counting and sampling algorithms that both run in time $O((n/\eps)^2 \log^3 (n/\eps) \log (1/\epsilon))$.
 \end{corollary}

As with independent sets, the extra factor in the running time for odd $q$ comes from the fact that the ground states (colorings in which one side of the bipartition is assigned $\lceil q/2 \rceil$ colors and the other side $\lfloor q/2 \rfloor$ colors) are exactly symmetric only if $q$ is even.

Finally, we remark that the approximate counting algorithms for these applications based on truncating the cluster expansion can run faster than $n^{O(\log \Delta)}$ if the parameters (expansion, fugacity, inverse temperature) are high enough (see \cite[Theorem 8]{JKP2}), but the sampling algorithms derived from this approach will not match the $\tilde O(n)$ or $\tilde O(n^2)$ sampling algorithms we obtain here.

\subsection{Comparison to spin Glauber dynamics}
\label{sec:firstcomparison}

A very natural idea to sample at low temperatures (large $\beta$ for the Potts model, large $\lam$ for the hard-core model) is to use a single-spin update Markov chain like the Glauber dynamics, but to start in one of the ground states of the model chosen at random.  For example, pick one of the $q$-colors with equal probability then start the Potts model Glauber dynamics in the monochromatic configuration with that color.  The intuition is that the Glauber dynamics will mix well within the portion of the state space close to the chosen ground state, and the randomness in the choice of ground state will ensure that an accurate sample from the full measure is obtained.  Analyzing this algorithm was suggested in~\cite{HPR} and~\cite{JKP}. 

While we are not yet able to show that this algorithm succeeds, we make partial progress.  We show that Glauber dynamics, restricted to remain in a portion of the state space, mixes rapidly (in polynomial time).  It is easiest to state our result for the ferromagnetic Potts model.  

For a ground state color $g \in [q]$ and an integer $M$, 
let $\Omega^{g}_{M}(G)$ be the set of $q$-colorings of the vertices of $G$ so that every connected component of $G$ colored with the palette of colors $[q] \setminus g$ is of size at most $M$.  The set  $\Omega^{g}_{M}(G)$ consists of colorings that come from the valid polymer configurations from Example~\ref{egPotts} above.   In~\cite{JKP} it is shown that for  an appropriate choice of $M$, the set $\{  \Omega^{g}_{M}(G)\}_{g \in [q]}$ forms an ``almost partition'' of the set of all colorings, in that the  weight of both the overlap of the almost partition and the set of colorings uncovered by the almost partition is at most $\epsilon$ under the conditions of Theorem~\ref{thmPotts}. In particular, an $\epsilon$-approximate sample from the Potts model restricted to  $\Omega^{g}_{M}(G)$ for $M = O(\log(n/\eps))$ is enough (by symmetry) to obtain a $(q\epsilon)$-approximate sample from the Potts distribution $\mu_{G,\beta}$ (cf. Lemma~\ref{lem:f34f34}).  Using Markov chain comparison, we show in Section~\ref{sec:f3f344} that this can be done using the usual spin Glauber dynamics restricted to remain in  $\Omega^{g}_{M}(G)$. 

\begin{restatable}{theorem}{thmPottsGlauber}
\label{thmPottsGlauber}
Suppose $q\geq 2$, $\Delta \geq 3$ are integers and $\alpha>0$ is a real. Let $\beta \geq \frac{5+3\log((q-1) \Delta)}{\alpha}$ be a real number and $g \in [q]$. Then, for any $n$-vertex $\alpha$-expander graph $G$ of maximum degree $\Delta$ and any $\epsilon \in (0, 1)$, for $M = O(\log(n/\eps))$ the Glauber dynamics restricted to  $\Omega^{g}_{M}(G)$ has mixing time $\Tmix(\eps)$ polynomial in $n$ and $1/\eps$. 
\end{restatable}

We remark that the polynomial bound in Theorem~\ref{thmPottsGlauber} depends on $q,\Delta,\alpha$ exponentially, see the relevant Theorem~\ref{thm:rglaubermix} and Section~\ref{sec:f3f344} for details. Theorem~\ref{thmPottsGlauber} shows that despite exponentially slow mixing of the Glauber dynamics on the full state space~\cite{bordewich2016mixing}, it can still be used to obtain a polynomial-time approximate sampling algorithm. We leave for future work two important extensions that would complete the picture: 1) showing that \textit{unrestricted} Glauber dynamics starting from a well chosen configuration works 2) lowering the running time to $O(n \log n)$ from the large polynomial we obtain in the theorem.  

In Section~\ref{secCompare}, we state a general theorem (Theorem~\ref{thm:rglaubermix}) comparing the polymer model dynamics to spin model dynamics  as well as a specific result for the hard-core model (Theorem~\ref{thmHCGlauber}).

\section{Polymer models and Markov chains}
\label{secPolymer}

Here we compare various conditions on the weight functions of a polymer model, namely the Koteck\'{y}--Preiss~\cite{KP} condition and the polymer sampling condition, and show that the latter implies the former. Then, we  define the polymer Markov chain which we use to prove Theorems~\ref{thmPolyMix} and~\ref{thmPolySample}. 

\subsection{A comparison of the conditions on the weights}
\label{secConditionCompare}

Here we  show that the polymer sampling condition~\eqref{eqSampleCondition}
implies the well-known 
Koteck\'{y}--Preiss~\cite{KP} condition:
\begin{align*}
\sum_{\gamma' \nsim \gamma} e^{|\gamma'|} w_{\gamma'} \le |\gamma| .
\end{align*}

To see the implication, 
we use a lemma of Borgs, Chayes, Kahn, and Lov{\'a}sz.
\begin{lemma}[\cite{BCKL}]
\label{lem:BCKL}
Let $G$ have maximum degree $\Delta \geq 3$ and let $v \in V(G)$. The number of connected induced subgraphs of $G$ of size $k$ containing $v$ is at most $(e \Delta)^{k-1}$. 
\end{lemma}

Now consider a polymer model satisfying~\eqref{eqSampleCondition} with constant
$\PSC \ge 5 + 3 \log ((q-1)\Delta)$. Fix $\gamma \in \cC(G)$. We have that
\[
\sum_{\gamma' \nsim \gamma} e^{|\gamma'|} w_{\gamma'} \leq \sum_{v \in \gamma \cup \partial \gamma} \sum_{k \geq 1} \sum_{\gamma' \in \cC(G) ;\, |\gamma'| = k,~v \in \gamma'} e^k e^{-\tau k}.
\]
In order to account for all of the polymers that we sum over in the above, we consider the connected induced subgraphs of $G$ of size $k$ that contain $v$, and the assignments to them of $q-1$ colours. Using Lemma~\ref{lem:BCKL}, we therefore obtain that
\begin{align*} \sum_{\gamma' \nsim \gamma} e^{|\gamma'|} w_{\gamma'} 
&\leq |\gamma| (\Delta+1)\sum_{k \ge 1}  (e \Delta)^{k-1} (q-1)^k e^k e^{-\PSC k}  =\frac{ |\gamma| (\Delta+1)}{e \Delta}\sum_{k \ge 1}  (e \Delta)^{k} (q-1)^k e^k e^{-\PSC k}  \\
&\le \frac{|\gamma| (\Delta+1)}{e \Delta} \sum_{k \ge 1} e ^{-3k} \leq |\gamma|,
\end{align*}
so the 
Koteck\'{y}--Preiss  condition is satisfied. 
 
The Koteck\'{y}--Preiss  condition, in turn, implies the polymer 
mixing condition~\eqref{eqMixCondition} 
with $\theta=1/e$ since 
$e\cdot x \le e^x$ for $x \ge 1$. 
For the same reason (since $e^x$ gets much bigger than $x$), it is easy to see that the polymer mixing condition 
is weaker than the Koteck\'{y}--Preiss condition.

\subsection{The polymer Markov chain}
\label{secpolyMCsec}
For each $v\in V(G)$, let $\cA(v) = \{ \gamma \in \cC(G) : v\in{\gam} \}$ denote the collection of all polymers containing $v$ and let $a(v) = \sum_{\gamma \in \cA(v)}  w_{\gamma}$. 
By applying~\eqref{eqMixCondition} to the smallest $\gamma$ containing $v$ we have $a(v) \le  \theta < 1$ for all $v \in V(G)$. Define the probability distribution $\nu_v$ on $\cA(v) \cup \{\emptyset\}$ by $\nu_v(\gamma) = w_{\gamma}$ for $\gam\in\cA(v)$ and $\nu_v(\emptyset) = 1- a(v)$.   

The polymer dynamics on $\Omega$ are defined by the following transition rule from a configuration $\Gamma_t$ to a configuration $\Gamma_{t+1}$:

\medskip
\textbf{Polymer Dynamics}
\begin{enumerate}
\item Choose $v \in V(G)$ uniformly at random. Let $\gamma_v \in \Gamma_t \cap \cA(v)$ if $\Gamma_t \cap \cA(v) \neq \emptyset$ and let $\gamma_v = \emptyset$ otherwise. Note that $\gamma_v$ is well defined since $\Gamma_t$ can have at most one polymer containing $v$. 
\item Mutually exclusively do the following:
\begin{itemize}
\item With probability $\frac{1}{2}$, let $\Gamma_{t+1} = \Gamma_t \setminus \gamma_v$. 
\item With probability $\frac{1}{2}$, sample $\boldsymbol \gamma$ 
from $\nu_v$, set $\Gamma_{t+1} = \Gamma_{t} \cup \boldsymbol \gamma$ if this is in $\Omega$ and set $\Gamma_{t+1}= \Gamma_t$ otherwise.
\end{itemize}
\end{enumerate}

Note that the polymer dynamics are aperiodic, since there are self-loops, and irreducible since we can transition from any $\Gamma \in \Omega$ to any $\Gamma' \in \Omega$ (e.g., via the empty set). Since the polymer dynamics are finite, irreducible, and aperiodic, they are also ergodic. Next, we observe that the stationary distribution of the polymer dynamics is $\mu_G$ by checking detailed balance.  Note that each transition of the dynamics changes a configuration $\Gamma$ by at most one polymer $\gamma$; let $\Gamma' = \Gamma \cup \gamma$.  Then
\begin{align*}
\frac{\mu_G(\Gamma')}{\mu_G(\Gamma)} = \frac{ \prod_{\gamma' \in \Gamma'} w_{\gamma'}}{ \prod_{\gamma' \in \Gamma} w_{\gamma'}} &= w_{\gamma} =  \frac{ \frac{|\gam|}{n} \cdot \frac{1}{2} \cdot w_\gam }{ \frac{|\gam|}{n} \cdot \frac{1}{2} }= \frac{P_{\Gamma \to \Gamma'}}{P_{\Gamma' \to \Gamma}} \,,
\end{align*}
where $P$ is the transition matrix of the polymer dynamics, and so $\mu_G$ is the stationary distribution.

We now formally define the mixing time. If $\calM$ is an ergodic Markov chain with transition matrix~$P$ and stationary distribution~$\nu$ then
the mixing time of~$\calM$ from a state~$x$ is given by
$$T_x(\epsilon) = \min \{ t>0 \mid \mbox{
for all $t'\geq t$, $ \| P^{t'}(x,\cdot)-\nu(\cdot)\|_{TV} \leq \epsilon$}\},$$
where $\|\nu'-\nu\|_{TV}$ denotes the total variation distance between distributions~$\nu$ and~$\nu'$.
The mixing time of~$\calM$ is given by 
$\Tmix(\epsilon) = \max_x T_x(\epsilon)$.  We will write $\Tmix(\calM,\eps)$ below if we need to emphasize which Markov chain we refer to. 

\subsection{Proof of Theorems~\ref{thmPolyMix} and~\ref{thmPolySample}}

\thmPolyMix*
\begin{proof} 
We will show that under condition~\eqref{eqMixCondition} the mixing time of the polymer dynamics is $O(n \log (n/\eps))$ by applying the path coupling technique. 
We define a metric $D(\cdot, \cdot)$ on $\Omega$ by setting $D(\Gamma, \Gamma') = 1$ if $\Gamma' = \Gamma \cup \{\gamma\}$ or $\Gamma = \Gamma' \cup \{\gamma\}$
for a polymer~$\gamma$ and extending this as a shortest path metric; i.e., $D(\Gamma, \Gamma') = |\Gamma \triangle \Gamma'|$ for any $\Gamma,\Gamma'\in \Omega$ where $\triangle$ denotes the symmetric difference of two sets. 

Now suppose we couple two chains $X_t$ and $Y_t$ by attempting the same updates in both chains at each step. Suppose that $X_t = Y_t \cup\{\gamma\}$ for some polymer $\gam$. With probability $\frac{|\gam|}{n} \cdot \frac{1}{2}$ we pick $v\in {\gam}$ and remove $\gamma_v$ which yields $X_{t+1} = Y_{t+1} = X_t$.  On the other hand, we may attempt to add a polymer $\gamma' \nsim \gamma$ so that $Y_t \cup \{\gamma'\} \in \Omega$. That is, $X_{t+1} = X_t = Y_t\cup \{\gam\}$ and $Y_{t+1} = Y_t \cup\{\gamma'\}$. This occurs with probability $\frac{|\gam'|}{n} \cdot \frac{1}{2} \cdot w_{\gam'}$ and in this case $D(X_{t+1}, Y_{t+1}) \leq 2$. Putting these together we can bound
\begin{align*}
\E [ D(X_{t+1}, Y_{t+1})] &\le 1 + \frac{1}{2 n} \Big[ -|\gamma| + \sum_{\gamma' \nsim \gamma} |\gamma'| w_{\gamma'} \Big] \,.
\end{align*}
Using~\eqref{eqMixCondition} we have $\sum_{\gamma' \nsim \gamma} |\gamma'| w_{\gamma'} \le \theta |\gamma|$ and so
\begin{align*}
\E [ D(X_{t+1}, Y_{t+1})] &\le 1-|\gamma| \frac{1-\theta}{2 n}
\leq 1 - \frac{1-\theta}{2 n}.
\end{align*} 
By the path coupling lemma (see~\cite[Section 6]{DG}), and with $W$ denoting the diameter of $\Omega$ under $D(\cdot, \cdot)$, we have that  
the mixing time is at most $\log(W/\epsilon) 2n/(1-\theta) = O(n \log (n/\epsilon))$, using that $W\leq 2n$. This finishes the proof.
\end{proof}

To prove Theorem~\ref{thmPolySample} we will show that a single update of the polymer dynamics can be computed in constant expected time. Assume that the polymer model is computationally feasible and that the polymer sampling condition \eqref{eqSampleCondition}
holds with constant $\PSC \geq 5+3\log((q-1)\Delta)$. 
We will use the following algorithm. Let $r = \PSC - 2 -\log((q-1)\Delta) \ge 3 + 2 \log((q-1)\Delta)$ and let $\cA_k(v) = \{ \gamma \in \cA(v) : |\gamma| \le k\}$.

\medskip
\textbf{Single polymer sampler}
\begin{enumerate}
\item Choose $\mathbf k$ according to the following geometric distribution: for $k$ a non-negative integer, 
$$\Pr[\mathbf k = k] = (1-e^{-r}) e^{-r k} \, .$$
This gives $\Pr[ \mathbf k \ge k] = e^{-r k}$. 
\item Enumerate all polymers in $\cA_{\mathbf{k}}(v)$ and compute their weight functions. 
\item Mutually exclusively output $\gamma \in \cA_{\mathbf{k}}(v)$ with probability $w_{\gamma}  e^{r |\gamma|}$, and with all remaining probability output $\emptyset$.  In particular if $\mathbf k=0$, then  output $\emptyset$ with probability $1$. 
\end{enumerate}

In order to show that this algorithm has constant expected running time, we will require the following result on enumerating connected subgraphs of bounded degree graphs.

\begin{lemma}[\cite{PR} Lemma 3.7]
\label{lem:numsubgraphs}
Let $G$ have maximum degree $\Delta $ and let $v \in V(G)$.  There is an algorithm running in time $O(k^5 (e \Delta)^{2k})$ that outputs a list of all connected subgraphs of $G$ of size at most $k$ containing $v$. 
\end{lemma}

We now proceed to prove the following lemma.

\begin{lemma}
\label{lemOnePolySample}
Under the polymer sampling condition~\eqref{eqSampleCondition} the output distribution of the \textbf{single polymer sampler} is $\nu_v$. Further, assuming the polymer model is computationally feasible,  the expected running time of the sampler is constant.  
\end{lemma}

\begin{proof}

We first show that the probabilities $w_{\gamma} e^{r |\gamma|}$ sum to less than $1$, which shows the last step of the sampling algorithm is well defined.
Since $\PSC - r = 2 +  \log ((q-1)\Delta)$,
\begin{align*}
\sum_{\gamma \in A(v) } w_{\gamma} e^{r |\gamma|}  &\le  \frac {1}{2}\sum_{k\ge 1}  (e\Delta)^{k-1} {(q-1)}^k e^{-\PSC k + r k}= \frac{1}{2 e \Delta} \sum_{k \ge 1} e^{-k} < 1 \, . 
\end{align*}

We next show that the output of the algorithm has distribution $\nu_v$.  Given $\gamma \in \cA(v)$, to output $\gamma$ we must choose $\mathbf k \ge |\gamma|$.  This happens with probability $e^{-r |\gamma|}$ by the distribution of $\mathbf k$. Conditioned on choosing such a $\mathbf k$, the probability we output $\gamma$ is $w_{\gamma}  e^{r |\gamma|}$, and multiplying these probabilities together gives $w_{\gamma}$ as desired.  Since this is true for all $\gamma \in \cA(v)$, the output distribution is exactly $\nu_v$. 

Finally we analyze the expected running time assuming that the model is computationally feasible. To do this, we observe that by Lemma~\ref{lem:numsubgraphs}, conditioned on the event that $\mathbf k =k$ the enumeration step of the algorithm takes time $O(k^5 (e \Delta)^{2k})$, and the time to  determine which polymers are allowed and computing their weights is $O(k^c(q-1)^k (e \Delta)^{k-1}/2)$ for some $c>0$, since the polymer model is computationally feasible; here, the factor $k^c$ accounts for the time to determine whether a single polymer of size $k$ is `allowed' and to compute its weight. Therefore, the expected running time is 
\begin{align*}
&=O\Big(1+  \sum_{k \ge 1} \Pr[\mathbf k =k] \big( k^5 (e \Delta)^{2k} + k^c (e (q-1)\Delta)^k \big)   \Big) \\
&=O\Big( 1+ \sum_{k \ge 1} e^{-rk}\, k^c (e (q-1) \Delta)^{2k} \Big) =O\Big( 1+  \sum_{k \ge 1} k^c\, e^{-(\PSC'+1) k} \Big)= O(1) \,,
\end{align*}
where $\PSC' = \PSC - 5-3\log((q-1)\Delta) \geq 0$.
\end{proof}

Finally we prove Theorem~\ref{thmPolySample}.
\thmPolySample*

\begin{proof}
By Theorem~\ref{thmPolyMix}, there is there is an integer $C_1 > 1$ (independent of $n$) so that if we start with the empty configuration $\Gamma_0 = \emptyset$ and run the polymer dynamics, then $\Gamma_{C_1 \ceil{n \log(n/\epsilon)}}$ has distribution within $\eps/2$ total variation distance of $\mu_G$. By Lemma~\ref{lemOnePolySample}, there is an integer $C_2 > 1$ (independent of $n$) such that the expected number of steps required to perform one update of the polymer dynamics is at most $C_2$. To compute an $\epsilon$-sample from $\mu_G$, we repeat the following $\ceil{\log(2/\epsilon)}$ times, independently, and if no configuration is returned we return the empty configuration. Run the polymer dynamics for $3 C_1 C_2 \ceil{n \log(n/\epsilon)}$ steps starting from $\Gamma_0 = \emptyset$, and if at least $C_1 \ceil{n \log(n/ \epsilon)}$ updates of the polymer dynamics were executed, return $\Gamma_{C_1 \ceil{n \log(n/ \epsilon)}}$.
 
We next show that the probability that the algorithm does not timeout and return the empty configuration is at least $1 - \eps/2$, which therefore yields that the output distribution has total variation distance at most $\eps$ from $\mu_G$. Let $X$ denote the total number of steps required to execute $C_1 \ceil{n \log(n/\epsilon)}$ updates of the polymer dynamics, and note that $\E[X] \leq C_1 C_2 \ceil{n \log(n/\epsilon)}$. By Markov's inequality, it follows that $\Pr(X \geq 3 \E[X]) < 1/e$. Thus, the probability that $X \geq 3\E[X]$ for each of $\ceil{\log(2/\epsilon)}$ independent copies of $X$, is less than $(1/e)^{\log(2/\epsilon)} = \epsilon/2$.
\end{proof}

\section{Approximate counting  algorithm}
\label{secCounting}
In this section we show how to use a sampling oracle to approximately compute the partition function of the polymer model. One standard way is by self-reducibility. In~\cite{HPR} an efficient sampling algorithm for polymer models is derived from an efficient approximate counting algorithm by applying self-reducibility on the level of polymers. While we could apply polymer self-reducibility in the other direction to obtain counting algorithms from our sampling algorithm, here we  use the simulated annealing method instead (see \cite{BSVV2008, H2015, SVV2009annealing}) to obtain a faster implementation of counting from sampling.

Suppose that $(\cC(G),w)$ is a computationally feasible polymer model. Let $\rho$ be a parameter and define a weight function
\[
	w_\gam(\rho) = w_\gam e^{-\rho|\gam|}
\]
for all $\gam \in \cC(G)$. 
Then for each $\rho \ge 0$ this defines a computationally feasible polymer model $(\cC(G),w(\rho))$ on $G$, where setting $\rho=0$ recovers the original model $(\cC(G),w)$. If the original model $(\cC(G),w)$ satisfies the polymer sampling condition \eqref{eqSampleCondition}, then so does $(\cC(G),w(\rho))$ for every $\rho \geq 0$ as the weight function $w_\gam(\rho)$ is monotone decreasing in $\rho$. 

Given the graph $G$, we write the partition function of the polymer model $(\cC(G),w(\rho))$ as a function of $\rho$:
\[
Z(\rho) = Z(G; \rho) = \sum_{\Gamma\in\Omega} \prod_{\gam\in\Gamma} w_\gam(\rho) = \sum_{\Gamma\in\Omega} \prod_{\gam\in\Gamma} w_\gam e^{-\rho|\gam|}.
\]
The associated Gibbs distribution is denoted by $\mu_\rho = \mu_{G;\rho}$. 
Since  $\lim_{\rho \to \infty} w_\gam(\rho) = 0$, we have  $\lim_{\rho \to \infty} Z(\rho) = 1$
(only the empty configuration~$\Gamma$ contributes to this limit), and so we will use simulated annealing to interpolate between $Z(\infty)=1 $ and our goal $Z(0)$, assuming access to a sampling oracle for $(\cC(G),w(\rho))$ for all $\rho\geq 0$. To apply the simulated annealing method, roughly speaking,  we find a sequence of parameters $0=\rho_0 < \rho_1 <\dots < \rho_\ell < \infty$ called a \textit{cooling schedule} where $\ell\in \N^+$, and then estimate $Z(0)$ using the telescoping product
\[
	\frac{1}{Z(0)} = \frac{1}{Z(\rho_0)} = \frac{Z(\rho_1)}{Z(\rho_0)} \frac{Z(\rho_2)}{Z(\rho_1)} \cdots \frac{Z(\rho_\ell)}{Z(\rho_{\ell-1})} \frac{1}{Z(\rho_\ell)}.
\]
To estimate each term $Z(\rho_{i+1})/Z(\rho_i)$, we define independent random variables
\[
	W_i = \prod_{\gam\in\Gamma_i} \frac{w_\gam(\rho_{i+1})}{w_\gam(\rho_i)}, \qquad\text{where } \Gamma_i \sim \mu_{\rho_{i}}.
\]
It is straightforward to see that $\E[W_i] = Z(\rho_{i+1})/Z(\rho_i)$ (see Lemma~\ref{lem:W_i-moments}). Using the sampling oracle for $\mu_{\rho_{i}}$, we can sample $W_i$ for all $i$, and by taking the product we get an estimate for $1/Z(0)$. 

The key ingredient of simulated annealing is finding a good cooling schedule. There are nonadaptive schedules \cite{BSVV2008} that depend only on $n$, and adaptive schedules \cite{H2015, SVV2009annealing} that also depend on the structure of $Z(\cdot)$. Usually the latter leads to faster algorithms than the former. In this paper we will use a simple nonadaptive schedule: $\rho_i = i/n$ for $i = 0,\dots,\ell$ where $\ell = O(n\log (n/\eps))$. We will show that this cooling schedule already gives us a fast algorithm for the polymer model. The reason behind it is that the weight function $w_\gam(\rho)$ decays exponentially fast, and so 
(see Lemma~\ref{lem:Z(rho_ell)})
the partition function $Z(\rho_\ell)$ is bounded by a constant when $\rho_\ell = O(\log n)$, leading to a short cooling schedule.

Our algorithm is as follows.

\medskip
\textbf{Polymer approximate counting algorithm}
\begin{enumerate}
	\item Let $\rho_i = i/n$ for $i=0,1,\dots,\ell$ where $\ell=\left\lceil n \log(4e(q-1)\Delta n/\eps) \right\rceil$;
	\item For $j=1,\dots, m$ where $m = \left\lceil 64\eps^{-2} \right\rceil$:
	\begin{enumerate}
		\item For $0\leq i\leq \ell-1$:
		\begin{enumerate}[(i)]
			\item Sample $\Gamma_i^{(j)}$ from $\mu_{\rho_i}$;
			\item Let $W_i^{(j)} = \prod_{\gam\in\Gamma_i^{(j)}} e^{-|\gam|/n}$;
		\end{enumerate}
		\item Let $W^{(j)} = \prod_{i=0}^{\ell-1} W_i^{(j)}$;
	\end{enumerate}
	\item Let $\widehat{W} =  
	\frac{1}{m} \sum_{j=1}^{m} W^{(j)}$ and output $\widehat{Z} = 1/\widehat{W}$.
\end{enumerate}

Before proving Theorem~\ref{thmPolyCount}, we first present a few useful lemmas. 
We shall use $\rho_i = i/n$ for $0 \leq i \leq \ell$ as our cooling schedule and we further define $\rho_{\ell+1} = (\ell+1)/n$ though it does not appear in the algorithm. For $0 \leq i \leq \ell-1$ independently we define $\Gamma_i$ to be a random sample from $\mu_{\rho_i}$ and $W_i = \prod_{\gam\in\Gamma_i} e^{-|\gam|/n}$. Finally, we let $W = \prod_{i=0}^{\ell-1} W_i$.

\begin{lemma}\label{lem:W_i-moments}
	For $0\leq i\leq \ell-1$,
	\[
		\E[W_i] = \frac{Z(\rho_{i+1})}{Z(\rho_i)} \qquad\text{and}\qquad \E[W_i^2] = \frac{Z(\rho_{i+2})}{Z(\rho_i)}.
	\]
	Therefore,
	\[
		\E[W] = \frac{Z(\rho_\ell)}{Z(0)} \qquad\text{and}\qquad \E[W^2] = \frac{Z(\rho_\ell)Z(\rho_{\ell+1})}{Z(0)Z(\rho_1)}.
	\]
\end{lemma}

\begin{proof} 
In the proof, we use $W_i(\Gamma_i)$ to denote
$\prod_{\gam\in\Gamma_i} \frac{w_\gam(\rho_{i+1})}{w_\gam(\rho_i)}$.
	We deduce from the definition of $W_i$ that
	\begin{align*}
	\E[W_i] &= \sum_{\Gamma_i\in\Omega} \mu_{\rho_i}(\Gamma_i) W_i(\Gamma_i) 
	= \frac{1}{Z(\rho_i)} \sum_{\Gamma_i\in\Omega} \prod_{\gam\in\Gamma_i} w_\gam e^{-i|\gam|/n} \prod_{\gam\in\Gamma_i} e^{-|\gam|/n}\\ 
	&= \frac{1}{Z(\rho_i)} \sum_{\Gamma_i\in\Omega} \prod_{\gam\in\Gamma_i} w_\gam e^{-(i+1)|\gam|/n} = \frac{Z(\rho_{i+1})}{Z(\rho_i)}
	\end{align*}
	and that
	\begin{align*}
	\E[W_i^2] &= \sum_{\Gamma_i\in\Omega} \mu_{\rho_i}(\Gamma_i) W_i(\Gamma_i)^2 
	= \frac{1}{Z(\rho_i)} \sum_{\Gamma_i\in\Omega} \prod_{\gam\in\Gamma_i} w_\gam e^{-i|\gam|/n} \prod_{\gam\in\Gamma_i} e^{-2|\gam|/n}\\ 
	&= \frac{1}{Z(\rho_i)} \sum_{\Gamma_i\in\Omega} \prod_{\gam\in\Gamma_i} w_\gam e^{-(i+2)|\gam|/n} = \frac{Z(\rho_{i+2})}{Z(\rho_i)}.
	\end{align*}
	Since $W_0,\dots,W_{\ell-1}$ are mutually independent, we obtain
	\[
		\E[W] = \prod_{i=0}^{\ell-1} \E[W_i] = \prod_{i=0}^{\ell-1} \frac{Z(\rho_{i+1})}{Z(\rho_i)} = \frac{Z(\rho_\ell)}{Z(\rho_0)}
	\]
	and
	\[
		\E[W^2] = \prod_{i=0}^{\ell-1} \E[W_i^2] = \prod_{i=0}^{\ell-1} \frac{Z(\rho_{i+2})}{Z(\rho_i)} = \frac{Z(\rho_\ell) Z(\rho_{\ell+1})}{Z(\rho_0) Z(\rho_1)}. \qedhere 
	\]
\end{proof}

\begin{lemma}\label{lem:Z(rho_ell)}
	Suppose that $w_\gam\leq 1$ for all $\gam\in\cC(G)$. 
	Then we have
	\[
		1 \leq Z(\rho_\ell) \leq e^{\eps/2}.
	\]
\end{lemma}

\begin{proof} 
	It is trivial that $Z(\rho_\ell) \geq 1$ since $\emptyset \in \Omega$ has weight $1$. Meanwhile, we have the crude bound 
	\[
	Z(\rho_\ell) \leq \prod_{\gam\in\cC(G)} \left( 1+w_\gam \e^{-\ell|\gam|/n} \right).
	\]
	We then deduce that
	\begin{align*}
	\log(Z(\rho_\ell)) &\leq \sum_{\gam\in\cC(G)} w_\gam e^{-\ell|\gam|/n} 
	\leq \sum_{v\in V} \sum_{k \geq 1} \sum_{\substack{\gamma \in \cC(G):\\ |\gamma| = k,\; v\in{\gam}}} \left( e^{-\ell/n} \right)^k\\
	&\overset{(a)}{\leq} n \,\sum_{k\geq 1} \left( \frac{e(q-1)\Delta}{e^{\ell/n}} \right)^k 
	\overset{(b)}{\leq} n \,\sum_{k\geq 1} \left( \frac{\eps}{4n} \right)^k \leq n \cdot \frac{2\eps}{4n} = \frac{\eps}{2}
	\end{align*}
	where (a) follows from Lemma~\ref{lem:BCKL} and (b) from $\ell \geq n\log(4e(q-1)\Delta n/\eps)$.
\end{proof}

\begin{lemma}\label{lem:Z1/Z0}
	We have
	\[
		\frac{Z(\rho_1)}{Z(0)} \geq \frac{1}{e} \qquad\text{and}\qquad \frac{Z(\rho_{\ell+1})}{Z(\rho_\ell)} \leq 1.
	\]
\end{lemma}
\begin{proof}
	Since the weight function $w_\gam(\rho)$ is decreasing in $\rho$, the partition function $Z(\rho)$ is also decreasing, which implies $Z(\rho_{\ell+1}) \leq Z(\rho_\ell)$. 
	On the other hand, recalling Lemma~\ref{lem:W_i-moments}, we have
	\[
	\frac{Z(\rho_1)}{Z(0)} = \E[W_0] = \E \bigg[ \prod_{\gam \in \Gamma_0} e^{-|\gam|/n} \bigg]
	\]
	where $\Gamma_0$ is sampled from $\mu_{\rho_0}$. 
	Notice that for any $\Gamma_0\in\Omega$ we have
	\[
	\prod_{\gam \in \Gamma_0} e^{-|\gam|/n} 
	= \exp\bigg( - \frac{1}{n} \sum_{\gam\in\Gamma_0} |\gam| \bigg) \geq \frac{1}{e}.
	\]
	Thus, the lemma follows.
\end{proof}

We are now ready to prove Theorem~\ref{thmPolyCount} which we restate for convenience.

\thmPolyCount*

\begin{proof} 
	We first assume that we have access to an exact sampler $\cSexact$ that samples from $\mu_{\rho}$ for all $\rho\geq 0$. 
Using this sampler in the Polymer approximate counting algorithm, we find that,
	  for each $j$ and each $i$, $\Gamma_i^{(j)}$ is an exact sample from the distribution $\mu_{\rho_i}$ and hence $W_i^{(j)}$ is an exact sample of $W_i$, independently for every $j$ and $i$. Thus, $W^{(j)}$ is a sample of $W$ independently for every $j$, and $\widehat{W}$ is the sample mean of $W^{(j)}$'s. We deduce from Lemmas~\ref{lem:W_i-moments} and \ref{lem:Z(rho_ell)} that
	\[
		(1+\eps/2) \E[W] \leq \frac{e^{\eps/2} Z(\rho_\ell)}{Z(0)} \leq \frac{e^{\eps}}{Z(0)}
	\]
	and
	\[
		(1-\eps/2) \E[W] \geq \frac{e^{-\eps} Z(\rho_\ell)}{Z(0)} \geq \frac{e^{-\eps}}{Z(0)}
	\]
	where we use $1+\eps/2 \leq e^{\eps/2}$ and $e^{-\eps} \leq 1-\eps/2$ for all $0<\eps<1$. 
	Then 
	\[
		\Pr\left( \frac{e^{-\eps}}{Z(0)} \leq \widehat{W} \leq \frac{e^{\eps}}{Z(0)} \right) \geq \Pr\left( \left|\widehat{W} - \E[W]\right| \leq (\eps/2)\E[W] \right).
	\]
	By Chebyshev's inequality we have
	\begin{align*}
		\Pr\left( \left|\widehat{W} - \E[W]\right| \geq (\eps/2)\E[W] \right) \leq \frac{4\,\mathrm{Var}(W)}{\eps^2 m \left(\E[W]\right)^2} 
		\leq \frac{4(e-1)}{\eps^2 m} \leq \frac{1}{8}
	\end{align*}
	where the second to last inequality follows from Lemmas~\ref{lem:W_i-moments} and \ref{lem:Z1/Z0}: 
	\[
		\frac{\mathrm{Var}(W)}{\left(\E[W]\right)^2} = \frac{\E[W^2]}{\left(\E[W]\right)^2} - 1 = \frac{Z(0)}{Z(\rho_1)} \frac{Z(\rho_{\ell+1})}{Z(\rho_\ell)} - 1 \leq e-1.
	\]
	Thus, we deduce that
	\[
	\Pr\left( e^{-\eps} Z(0) \leq \widehat{Z} \leq e^{\eps} Z(0) \right) = 
	\Pr\left( \frac{e^{-\eps}}{Z(0)} \leq \hat{W} \leq \frac{e^{\eps}}{Z(0)} \right) \geq \frac{7}{8},
	\]
so the error probability is at most~$1/8$. 
Note that the number of samples that we used is~$\ell m$.
	
	Now we replace the exact sampling oracle $\cSexact$ by an approximate one. 
	For every $\rho \geq 0$, the polymer model $(\cC(G),w(\rho))$ is computationally feasible and satisfies the polymer sampling condition \eqref{eqSampleCondition}. 
	Thus, 
for any $\rho\geq 0$,	
	Theorem~\ref{thmPolySample} gives a randomized algorithm $\mathcal{S}$ that outputs a $1/(8 \ell m)$-approximate sample from $\mu_{\rho}$. We then couple $\mathcal{S}$ and $\cSexact$ optimally and run the algorithm with both $\mathcal{S}$ and $\cSexact$ simultaneously, so that for any $\rho\geq 0$ samples from $\mathcal{S}$ and $\cSexact$ for $\mu_{\rho}$ coincide with probability at least $1-1/(8\ell m)$.
	Let $\mathcal{B}$ be the event that at least one of the $\ell m$ samples from $\mathcal{S}$ in the algorithm does not couple with that from $\cSexact$. Then a union bound yields $\Pr(\mathcal{B}) \leq 1/8$. 
	Let $\mathcal{F}$ be the event that the algorithm using $\cSexact$ fails. From our argument before we see that $\Pr(\mathcal{F}) \leq 1/8$. 
	Note that if neither of $\mathcal{B}$ and $\mathcal{F}$ happens, then the algorithm with $\mathcal{S}$ will output a desired estimate. 
	Hence, we conclude from the union bound that the algorithm with $\mathcal{S}$ fails with probability at most
	\[
	\Pr(\mathcal{B}) + \Pr\left( \mathcal{F} \right) \leq \frac{1}{8} + \frac{1}{8} = \frac{1}{4}.
	\]
	
	Finally, we consider the running time of our algorithm.
	By Theorem~\ref{thmPolySample} the running time of step 2(a)(i) is $O(n\log(8\ell m n) \log(8 \ell m) ) = O(n\log^2 (n/\eps))$, and for step 2(a)(ii) the running time is $O(n)$.  
	Thus, the  running time of the algorithm is upper bounded by $\ell m \cdot O(n\log^2(n/\eps)) = O((n/\eps)^2 \log^3(n/\eps))$. 
\end{proof}

\section{Applications}
\label{secApply}

Here we apply our results on subset polymer models to several approximate counting and sampling problems at low temperatures.

\subsection{Ferromagnetic Potts model}\label{sec:appPotts}
In this section, we prove Theorem~\ref{thmPotts} for the Potts model. Throughout this section, we will work under the assumptions/conditions of Theorem~\ref{thmPotts}.
That is, we fix a real number $\alpha>0$, integers $q \geq 2$ and $\Delta \geq 3$ and a real number
$\beta\geq \frac{5+3\log((q-1)\Delta)}{\alpha}$.
We let $\mathcal{G}$ be the class of $\alpha$-expander graphs~$G$ with maximum degree at most~$\Delta$.

Consider the polymer model defined in Example~\ref{egPotts} 
on an $n$-vertex graph $G\in \mathcal{G}$
with $M = n/2$ and ground state color $g \in [q]$. We will use $\mathcal{C}^g=\mathcal{C}^g(G)$ to denote the polymers and $w_\gamma^g$ to denote the weight of a polymer $\gamma\in \mathcal{C}^g$; recall that $w_\gamma^g = e^{-\beta B(\gamma)}$, where $B(\gamma)$ counts the number of external edges of $ \gamma$ plus the number of bichromatic internal edges.  Let $Z^g(G)$ be the partition function of the polymer model $(\mathcal{C}^g(G), w^g)$.    
\begin{lemma}\label{SampleConditionPotts}
Under the conditions of Theorem~\ref{thmPotts}, the polymer model 
$(\mathcal{C}^g(\cdot), w^g,\mathcal{G})$ 
satisfies the polymer sampling condition \eqref{eqSampleCondition} with $\tau=\alpha \beta$.
\end{lemma}
\begin{proof}
Since 
every $G\in \mathcal{G}$
is an $\alpha$-expander, for $\gamma\in \mathcal{C}^g$ we have $B(\gamma) \ge \alpha|\gamma|$ and hence $w_\gamma^g \leq e^{-\tau|\gamma|}$.
\end{proof}

The following lemma is from \cite{JKP2}.    
\begin{lemma}[{\cite[Lemma 12]{JKP2}}]
\label{lemPottsApprox}
For any $n$-vertex $\alpha$-expander graph $G$ and $\beta\ge 2 \log (eq)/\alpha$, $q Z^g(G)$ is an $e^{-n}$-approximation of the Potts partition function $Z_{G,\beta}$.
\end{lemma}

We are now ready to prove Theorem~\ref{thmPotts}.
\thmPotts*

\begin{proof} 
Let $\mathcal{G}$ be the class of $\alpha$-expander graphs of maximum degree at most~$\Delta$.
Clearly, the polymer models 
$(\mathcal{C}^g(\cdot), w^g,\mathcal{G})$  
are computationally feasible. By Lemma~\ref{SampleConditionPotts}, the models also satisfy the polymer sampling condition and therefore Theorems~\ref{thmPolySample} and~\ref{thmPolyCount} apply. 
Consider any $n$-vertex graph $G\in \mathcal{G}$. Since
$\beta\geq\frac{ 5+ 3 \log ((q-1)\Delta)   }{\alpha   } > \frac{2 \log (eq)}{\alpha}$, 
Lemma~\ref{lemPottsApprox} applies 
to~$G$.

For the sampling algorithm, we pick a color $g\in [q]$ uniformly at random and generate an $(\epsilon/q)$-approximate sample from the Gibbs measure associated to $Z^g(G)$ using the algorithm of Theorem~\ref{thmPolySample}, in time $O(n \log(n/\epsilon) \log(1/\epsilon))$. By Lemma~\ref{lemPottsApprox}, we conclude  that the resulting output is an $\epsilon$-approximate sample for the Potts model.

For the counting algorithm, we pick an arbitrary  $g\in [q]$ and produce using the algorithm of Theorem~\ref{thmPolyCount} a number $\hat{Z}$ in time 
$O((n/\epsilon)^2 \log^3(n/\epsilon))$, 
which is an  
$\epsilon/(2q)$-approximation 
to $Z^g(G)$ with probability $\geq 3/4$. By Lemma~\ref{lemPottsApprox}, we conclude  that $q\hat{Z}$ is an $\epsilon$-approximation for the partition function of the Potts model (with the same probability).
\end{proof}

\subsection{Hard-core model}\label{sec:apphardcore}
In this section, we prove Theorem~\ref{thmHC} for the hard-core  model.  

Suppose $G=(\Vpart{0},\Vpart{1},E)$ is an $n$-vertex bipartite $\alpha$-expander graph of maximum degree $\Delta$. We will consider the  hard-core model on $G$ at sufficiently large fugacities $\lambda$.  There are two relevant ground states corresponding to the two parts of $G$, one is the independent set given by $\Vpart{0}$  and the other is given by $\Vpart{1}$.  We will capture deviations from the two ground states using the ``even'' and ``odd'' polymer models of Jenssen, Keevash and Perkins~\cite{JKP2}. We remark that similar models were considered independently by Liao,  Lin, Lu, and Mao \cite{liao2019counting}.

For $i\in \{0,1\}$, we say a set $S \subseteq \Vpart{i}$ is \textit{small} if $|S| \le |\Vpart{i}|/2$.  In particular, Definition~\ref{defbipartiteexpander} requires that small sets expand.   

Following \cite{JKP2}, we will define a polymer model $(\Cpart{i}(G^2),\wpart{i})$; note that the host graph\footnote{$G^2$ is the graph on vertex set $V(G)$, where two vertices are connected if their distance in $G$ is at most 2.} is $G^2$,  rather than $G$. 
 The set $\Cpart{i} = \Cpart{i}(G^2)$ of allowed polymers consists of all small sets $\gamma \subseteq \Vpart{i}$ which are connected subgraphs in $G^2$. The set of spins is $\{0,1\}$ and the ground state spin for a vertex $v$ is $g_v=0$ if $v\in \Vpart{i}$, and $g_v=1$ if $v\in \Vpart{1-i}$; the spin assignment $\sigma_\gamma$ for a polymer $\gamma$ gives the spin $1$ to each $v \in \gamma$.   The weight of a polymer $\gamma\in \Cpart{i}$ is defined as 
\begin{equation}\label{eq:4v4tbtyb5y}
\wpart{i}_\gamma = \frac{\lambda^{| \gamma|}}{(1+\lambda)^{N_G(\gamma)}}, 
\end{equation} 
where we recall that $N_G(\gamma)$ denotes the set of vertices in $G$ which are adjacent to some vertex in $\gamma$. The key observation behind the definition of the weights 
is that for a set $\Gamma$ of compatible polymers   from $\Cpart{i}$, the contribution to $Z_{G,\lambda}$ of all independent sets $I$ with $I\cap \Vpart{i}= \bigcup_{\gamma \in \Gamma} \gamma$ is exactly 
\begin{align*}
(1+\lambda)^{|\Vpart{1-i}|} \prod_{\gamma \in \Gamma} \wpart{i}_\gamma \, ,
\end{align*}
see \cite[Proof of Lemma 19]{JKP2} for more details.

 Let $Z^i(G)$ denote the partition function of the polymer model $(\Cpart{i}, \wpart{i})$ (where two polymers are compatible if their distance in the host graph $G^2$ is at least 2). Using that $G$ is an $\alpha$-expander, we have the following lemma from \cite{JKP2}.
\begin{lemma}[{\cite[Lemma 19]{JKP2}}]
\label{lemHCApprox1}
For any $\lam \ge e^{11/\alpha}$ and any $n$-vertex graph $G=(\Vpart{0},\Vpart{1},E)$ which is a bipartite $\alpha$-expander, the number 
\[(1+\lambda)^{|\Vpart{1}|} Z^0(G) + (1+\lambda)^{|\Vpart{0}|} Z^1(G)\]
is an $e^{-n}$-approximation of the hard-core partition function $Z_{G,\lambda}$.
\end{lemma}
In particular, \cite[Lemma 17]{JKP2} shows that $(1+\lambda)^{|\Vpart{1}|} Z^0(G) + (1+\lambda)^{|\Vpart{0}|} Z^1(G)$ counts the contribution to  $Z_{G,\lambda}$ of every independent set $I$ of $G$, but some independent sets are double counted: those independent sets $I$  for which the $2$-connected components of $\Vpart{0} \cap I$ and $\Vpart{1} \cap I$ are all small. We call these independent sets \textit{sparse}.   The proof of \cite[Lemma 19]{JKP2} shows that the relative contribution to $Z_{G,\lambda}$ of sparse independent sets is at most $e^{-n}$. 

We are now ready to prove Theorem~\ref{thmHC}.
\thmHC*
\begin{proof} 
First note that $\lambda \geq (3\Delta)^{6/\alpha}\geq 9^{6/\alpha}>e^{11/\alpha}$, so Lemma~\ref{lemHCApprox1} applies. 
Let $\mathcal{G}$ denote the set of host graphs $G^2$ corresponding to bipartite $\alpha$-expanders~$G$ of maximum degree~$\Delta$.
Noting that the polymer models 
$(\mathcal{C}^i(\cdot),w^i,\mathcal{G})$
are computationally feasible, we  verify the polymer sampling condition~\eqref{eqSampleCondition} for them. Fix arbitrary $i\in\{0,1\}$.   As in   \cite[Section 4.2]{JKP}, we have the bound
\[\wpart{i}_\gamma=\frac{\lambda^{| \gamma|}}{(1+\lambda)^{N_G(\gamma)}}\leq \frac{\lambda^{| \gamma|}}{(1+\lambda)^{(1+\alpha)|\gamma|}}\leq \lambda^{-\alpha|\gamma|},\]
so, using that $\lambda \geq (3\Delta)^{6/\alpha}$, we have that the models satisfy the polymer sampling condition with $\tau=\alpha\log \lambda\geq 6\log (3\Delta)\geq 5+3\log \Delta^2$.  Therefore, we may also apply Theorems~\ref{thmPolySample} and~\ref{thmPolyCount}.

For the counting algorithm, we apply Theorem~\ref{thmPolyCount}. Namely, by taking the median of $O(\log(1/\epsilon))$ trials,  we can obtain $\hat{Z}^0$ and $\hat{Z}^1$ which are $(\eps/32)$-approximations to $Z^0(G)$ and $Z^1(G)$, respectively, with probability at least $1-\epsilon/32$. Let $\mathcal{E}$ be the event that $\hat{Z}^0$ and $\hat{Z}^1$ are indeed $(\eps/32)$-approximations to $Z^0(G)$ and $Z^1(G)$. Conditioned on $\mathcal{E}$, the number
\[\hat{Z}=(1+\lambda)^{|\Vpart{1}|} \hat{Z}^0 + (1+\lambda)^{|\Vpart{0}|} \hat{Z}^1\]
is an $(\eps/32)$-approximation to the number $A=(1+\lambda)^{|\Vpart{1}|} Z^0(G) + (1+\lambda)^{|\Vpart{0}|} Z^1(G)$. By Lemma~\ref{lemHCApprox1} and since $\epsilon\geq 4e^{-n}$,  $A$ is  an $(\epsilon/4)$-approximation to $Z_{G,\lambda}$ and hence $\hat{Z}$ is an $\eps$-approximation to $Z_{G,\lambda}$. Since $\mathcal{E}$ occurs with probability at least $1-\epsilon/16$, we obtain that $\hat{Z}$ is the desired approximation for the counting algorithm.

For the sampling algorithm,  let $\mathbf{i}$ be the random variable which takes the value 0 with probability $\frac{(1+\lambda)^{|\Vpart{1}|} \hat{Z}^0}{\hat{Z}}$ and the value 1 otherwise, where $\hat{Z}^0, \hat{Z}^1,\hat{Z}$ are the quantities computed earlier.  Then, use Theorem~\ref{thmPolySample} to obtain an $(\epsilon/8)$-approximate sample from the Gibbs distribution corresponding to the polymer model $(\mathcal{C}^{\mathbf{i}}(G),w^{\mathbf{i}})$, say  
$\hat{\Gamma}^{\mathbf{i}}$.  
Obtain then an independent set $\hat{I}$ by including into $\hat{I}$ each 
$v\in \Vpart{1-\mathbf{i}}\backslash N_G (  \bigcup_{\gamma \in \hat{\Gamma}^{\mathbf{i}}} \gamma)$
with probability $\frac{\lambda}{1+\lambda}$ and each vertex in 
$\bigcup_{\gamma \in \hat{\Gamma}^{\mathbf{i}}} \gamma$
(with probability 1). We claim that the output distribution of $\hat{I}$ is $\epsilon$-close to the hard-core distribution $\mu_{G,\lambda}$.

To prove this, consider the random independent set $I$ obtained by repeating the same steps above but using instead perfectly accurate computations, i.e., pick $i=0$ with probability $\frac{(1+\lambda)^{|\Vpart{1}|} Z^0(G)}{A}$ and the value 1 otherwise, then, sample (perfectly) 
$\Gamma^i$
from the Gibbs distribution corresponding to the polymer model $(\mathcal{C}^{i}(G),w^{i})$, and then obtain the independent set $I$ by including into $I$ each 
 $v\in \Vpart{1-i}\backslash N_G (  \bigcup_{\gamma \in {\Gamma}^{{i}}} \gamma)$ 
with probability $\frac{\lambda}{1+\lambda}$ and each vertex in 
 $\bigcup_{\gamma \in {\Gamma}^{{i}}} \gamma$
(with probability 1). Then, if $I$ is not sparse,  $I$ is generated with probability $\lambda^{|I|}/A$
 (cf. the observation below \eqref{eq:4v4tbtyb5y}).  On the other hand, if $I$ is sparse, then $I$ is generated with probability $2\lambda^{|I|}/A$.  But by Lemma~\ref{lemHCApprox1} and the remark following, the total variation distance between the distribution of $I$ and the hard-core distribution $\mu_{G,\lambda}$ is bounded by the relative weight of the sparse independent sets,  
 which, by Lemma~\ref{lemHCApprox1}, is at most $e^{-n} \leq \epsilon/4$.

We next  observe that, conditioned on the event $\mathcal{E}$ (i.e., that $\hat{Z}^0$ and $\hat{Z}^1$ are $(\eps/32)$-approximations to $Z^0(G)$ and $Z^1(G)$), there is a coupling between $\hat{I}$ and $I$ such that $\hat{I}=I$ with probability at least $1-\epsilon/4$. Indeed, the total variation distance between $\mathbf{i}$ and  $i$ is at most $e^{\epsilon/16}-1\leq \epsilon/8$ and hence there is a coupling of $\mathbf{i}$ with $i$ so that $\mathbf{i}=i$ with probability at least $1-\epsilon/8$. Analogously, there is a coupling of 
$\hat{\Gamma}^{\mathbf{i}}$ with $\Gamma^{i}$ so that  $\hat{\Gamma}^{\mathbf{i}}=\Gamma^{i}$ 
with probability at least $1-\epsilon/8$. Since $\mathcal{E}$ occurs with probability at least $1-\epsilon/16$, it follows that  the overall total variation distance between $\hat{I}$ and $I$ is at most $\eps/2$. 

Hence, the output distribution of $\hat{I}$ is $\epsilon$-close to the hard-core distribution $\mu_{G,\lambda}$, finishing the proof of Theorem~\ref{thmHC}.
\end{proof}

\section{Comparison to spin Glauber dynamics}
\label{secCompare}
In this section, we derive results for spin Glauber dynamics, restricted to appropriate sets in the state space,  based on our results above (using fairly standard Markov chain comparison techniques). We start with the general framework of subset polymer models and obtain Theorem~\ref{thm:rglaubermix}, which is  then applied to the ferromagnetic Potts and hard-core models.
\subsection{Restricted Glauber dynamics for polymer models}
Here, we define the restricted Glauber dynamics for subset polymer models, and show the upcoming Theorem~\ref{thm:rglaubermix} which bounds its mixing time under some appropriate conditions.

Consider a subset polymer model as in Section~\ref{sec:absolmod}.
There is a natural map  $f: \Omega \to \{0,1,\hdots,q-1\}^{V(G)}$ between allowed polymer configurations and spin configurations,  given by $f(\Gamma)_v = \sigma_{\gamma}(v)$ if $\gamma \in \Gamma$ and $v \in  \gamma$ and $f(\Gamma)_v = g_v$ if $v \notin \cup_{\gamma \in \Gamma}  \gamma$. Let $\Omega_{\mathrm{spin}}=f(\Omega)$ be the spin configurations obtainable as images of the map $f$.   It will be helpful to consider the inverse map $f^{-1}$ and extend its domain to all $\sigma \in \{0,1,\hdots,q-1\}^{V(G)}$, so that $f^{-1}(\sigma)$ is the polymer configuration consisting of polymers that are connected components of vertices which do not receive their ground state spin; note that the range of the extended $f^{-1}$ is not limited to $\Omega$ anymore.

Restricted Glauber dynamics is defined as follows,
  starting from $\Gamma_t \in \Omega$.

\begin{enumerate}
\item Choose $v \in V(G)$ and $s\in \{0,\ldots,q-1\}$ uniformly.

\item $\Gamma'$ is formed from $\Gamma_t$ by assigning $v$ to spin~$s$
(formally, by letting $\sigma = f(\Gamma_t)$, forming $\sigma'$ from $\sigma$ by assigning $v$ to spin~$s$, and
finally letting $\Gamma' = f^{-1}(\sigma')$).

\item If $\Gamma' \in \Omega$ let $p =  \min(1,w(\Gamma')/w(\Gamma))$.  
\begin{itemize}
\item With probability $p$, $\Gamma_{t+1} = \Gamma'$.
\item With probability $1-p$,  $\Gamma_{t+1} = \Gamma_t$.
\end{itemize}
\item If $\Gamma' \notin \Omega$ then $\Gamma_{t+1} = \Gamma_t$.
\end{enumerate}

We will use the Markov chain comparison technique to show that  the restricted Glauber dynamics is rapidly mixing.
To do this, we need a mild condition on the set of allowed polymers~$\cC(G)$. A polymer model is said to be \emph{single-update-compatible} if, for every size-$k$ polymer~$\gamma\in \cC(G)$, there is an ordering $v_1,\ldots,v_k$ of the vertices in ${\gamma}$ such that, for all $i\in [k]$, the set $S_i = \{ v_1,\ldots,v_i\}$ induces a connected subgraph of $\gamma$ and  we have that $(S_i, \sigma_\gamma|_{S_i})$ is a valid polymer itself, i.e., $(S_i, \sigma_\gamma|_{S_i}) \in \cC(G)$.

We will use the comparison method of Diaconis and Saloff-Coste~\cite{DS2,DS1} as applied to mixing times by Randall and Tetali~\cite{RT}. In order to avoid discussion of eigenvalues here, we use the version from Observation~13 of the survey paper~\cite{comparisonsurvey}. We first show that the restricted Glauber dynamics is a reversible ergodic Markov chain with stationary distribution $\mu_G$, which is easy to see from its definition.

\begin{lemma}
Let $G$ be a graph and let $(\cC(G), w)$ be a single-update-compatible polymer model. The restricted Glauber dynamics is ergodic and reversible with stationary distribution $\mu_G$.
\end{lemma}
\begin{proof}
The restricted Glauber dynamics is aperiodic since we remain in the same state with positive probability after performing an update. It is irreducible since we can transition from any $\Gamma \in \Omega$ to any $\Gamma' \in \Omega$ by adding and removing vertices one-by-one. This shows that the restricted Glauber dynamics is ergodic. To show that it is reversible and has stationary distribution $\mu_G$, we check detailed balance. Suppose $\Gamma,\Gamma'\in\Omega$ with $\Gamma\neq\Gamma'$ and $P(\Gamma,\Gamma')>0$ where $P$ is the transition matrix of the restricted Glauber dynamics. Then, 
	\[
		\frac{P(\Gamma,\Gamma')}{P(\Gamma',\Gamma)} = \frac{\frac{1}{n} \cdot \frac{1}{q} \cdot \min\{1,w(\Gamma')/w(\Gamma)\}} {\frac{1}{n} \cdot \frac{1}{q} \cdot \min\{1,w(\Gamma)/w(\Gamma')\}} = \frac{w(\Gamma')}{w(\Gamma)} = \frac{\mu_G(\Gamma')}{\mu_G(\Gamma)}.
	\]
	The lemma follows.
\end{proof}

We next give some standard definitions that will be used in our comparison proof. Let $\rGD$ denote the restricted Glauber dynamics and $P$ be its transition matrix. Let $\PD$ be the polymer dynamics and denote its transition matrix by $P'$. Define $E^*(\rGD)$ to be the set of pairs of configurations $(\Gamma,\Gamma')$ that can be achieved by one transition of the restricted Glauber dynamics; i.e., $E^*(\rGD) = \{(\Gamma,\Gamma')\in \Omega^2: P(\Gamma,\Gamma')>0\}$. Similarly, define $E^*(\PD) = \{(\Gamma,\Gamma')\in \Omega^2: P'(\Gamma,\Gamma')>0\}$ for the polymer dynamics. 

For every $(\Gamma,\Gamma') \in E^*(\PD)$, we define a path $\mathcal{P}_{\Gamma,\Gamma'}$ from $\Gamma$ to $\Gamma'$ to be a sequence of configurations such that every adjacent pair is a transition of the restricted Glauber dynamics; i.e, every adjacent pair of configurations is in $E^*(\rGD)$. 
For this, we assume that the polymer model is single-update-compatible (see Section~\ref{sec:firstcomparison}).
If $\Gamma = \Gamma'$, then the choice is easy --- we let $\mathcal{P}_{\Gamma,\Gamma'} = (\Gamma,\Gamma')$. 
Suppose instead that $\Gamma' = \Gamma \cup \gamma$ for some polymer $\gam \in \Omega$. Recall that there is a natural one-to-one mapping $f: \Omega\to \Omega_{\mathrm{spin}}$ between the set of all (polymer) configurations $\Omega$ and the set of spin configurations $\Omega_{\mathrm{spin}}$. Let $\sigma = f(\Gamma)$ and $\sigma' = f(\Gamma')$ be the corresponding spin configurations. 
If $\gamma$ has size $k$, let $v_1,\ldots,v_k$ 
be the ordering of vertices of~$\gamma$ from the definition of single-update-compatible
so that, for all $i\in [k]$, the polymer induced by vertices $v_1,\ldots,v_i$ is in $\cC(G)$. Let  $(\sigma = \sigma_0, \sigma_1,\dots,\sigma_k = \sigma')$ be the sequence of spin configurations such that each $\sigma_j$ is obtained from $\sigma_{j-1}$ by changing the spin of $v_j$ from $\sigma(v) = g_v$ to $\sigma'(v)$. The path $\mathcal{P}_{\Gamma,\Gamma'}$ is then defined to be $(f^{-1}(\sigma_0),\dots,f^{-1}(\sigma_k))$. If $\Gamma' = \Gamma \backslash \gam$ for some $\gam \in \Omega$, we can define the path $\mathcal{P}_{\Gamma,\Gamma'}$ in  a similar manner. Note that in both cases the length of the path is $|\mathcal{P}_{\Gamma,\Gamma'}| = k = |\gam|$. 

For every $(\Gamma_0,\Gamma_0') \in E^*(\rGD)$, the \textit{congestion} of the edge $(\Gamma_0,\Gamma_0')$ is defined to be
\[
A(\Gamma_0,\Gamma_0') = \frac{1}{\mu_G(\Gamma_0) P(\Gamma_0,\Gamma_0')} \sum_{\substack{(\Gamma,\Gamma') \in E^*(\PD):\\ \mathcal{P}_{\Gamma,\Gamma'} \ni (\Gamma_0,\Gamma_0')}} \mu_G(\Gamma) P'(\Gamma, \Gamma') |\mathcal{P}_{\Gamma,\Gamma'}|.
\]
The \textit{congestion} of the choice of paths is the quantity
\[
A = \max\limits_{(\Gamma_0,\Gamma_0') \in E^*(\rGD)} A(\Gamma_0,\Gamma_0').
\]
The following comparison lemma gives an upper bound on the mixing time of the restricted Glauber dynamics by the mixing time of the polymer dynamics.

\begin{lemma}[{\cite[Observation 13]{comparisonsurvey}}]
\label{lem:comparison}
Let $c_1 = \min_{\Gamma \in \Omega} P(\Gamma,\Gamma)$ and $c_2 = \min_{\Gamma \in \Omega} \mu_G(\Gamma)$. Then, for any $0<\eps <1$ we have
\[
\Tmix(\rGD,\eps) \le \max \left\{ A\left( \Tmix\left(\PD,\frac{1}{2\e}\right) +1 \right), \frac{1}{2c_1} \right\} \ln \frac{1}{\eps c_2}.
\]
\end{lemma}

We now proceed to establish the mixing-time of the restricted Glauber dynamics, which is the main result of this section.
We will apply this to both the hard-core model (on bipartite $\alpha$-expander graphs) and the ferromagnetic Potts model (on $\alpha$-expander graphs), for which we will define appropriate single-update compatible polymer models. Furthermore, in both of these applications, $M$ below will be logarithmic in $n/\epsilon$, giving polynomial mixing time for the restricted Glauber dynamics. 

\begin{theorem}\label{thm:rglaubermix}
Suppose that a polymer model $(\cC( \cdot ), w,\mathcal{G})$
satisfies the polymer mixing condition.
Consider a graph $G\in \mathcal{G}$ such that
$(\cC(G),w)$ 
 is single-update-compatible. Let $M = \max\{|\gam|: \gam\in\cC(G)\}$.  Suppose that, for every pair of configurations $\Gamma,\Gamma'\in\Omega$ whose corresponding spin configurations $f(\Gamma),f(\Gamma') \in\Omega_{\mathrm{spin}}$ differ at exactly one vertex, we have
\begin{equation}
\frac{1}{\eta} \leq \frac{\mu_G(\Gamma)}{\mu_G(\Gamma')} \leq \eta
\end{equation}
for some constant $\eta>1$. Then for any $0<\eps<1$, the restricted Glauber dynamics has mixing time
\[
\Tmix(\eps) \le O\left(M\eta^{M+1} n^2\log n \log(\eta/\eps) \right).
\]
\end{theorem}
\begin{proof}
By Lemma~\ref{lem:comparison}, it suffices to upper bound the congestion $A(\Gamma_0,\Gamma_0')$ for every $(\Gamma_0,\Gamma_0') \in E^*(\rGD)$ where
\[
A(\Gamma_0,\Gamma_0') = \sum_{\substack{(\Gamma,\Gamma') \in E^*(\PD):\\ \mathcal{P}_{\Gamma,\Gamma'} \ni (\Gamma_0,\Gamma_0')}} \frac{\mu_G(\Gamma)}{\mu_G(\Gamma_0)} \cdot \frac{P'(\Gamma, \Gamma')}{P(\Gamma_0,\Gamma_0')} \cdot |\mathcal{P}_{\Gamma,\Gamma'}|.
\]
If $\Gamma_0 = \Gamma_0'$, then for our choices of paths to get $(\Gamma_0,\Gamma_0') \in \mathcal{P}_{\Gamma,\Gamma'}$ we must have $\Gamma = \Gamma' = \Gamma_0 = \Gamma_0'$. It follows that
\[
A(\Gamma_0,\Gamma_0) = \frac{P'(\Gamma_0, \Gamma_0)}{P(\Gamma_0,\Gamma_0)} \leq \frac{1}{1/q} = q
\]
since $P(\Gamma_0,\Gamma_0) \geq 1/q$ by the update rule of the restricted Glauber dynamics.

Now suppose $\Gamma_0 \neq \Gamma_0'$. Let $\sigma_0 = f(\Gamma_0)$ and $\sigma_0' = f(\Gamma_0')$ be the corresponding spin configurations. Notice that $\sigma_0$ and $\sigma_0'$ differ at exactly one vertex, which we denote by $v$. If $\sigma_0(v) \neq g_v$ and $\sigma_0'(v) \neq g_v$ then no path $\mathcal{P}_{\Gamma,\Gamma'}$ would contain $(\Gamma_0,\Gamma_0')$ by our choice of paths, and thus $A(\Gamma_0,\Gamma_0) = 0$. Assume next that $\sigma_0(v) = g_v$ and $\sigma_0'(v) \neq g_v$. Then, if $(\Gamma_0,\Gamma_0') \in \mathcal{P}_{\Gamma,\Gamma'}$ for some $(\Gamma,\Gamma') \in E^*(\PD)$, we must have $\Gamma' = \Gamma \cup \gam$ for some polymer $\gam\in \Omega$ and also $v\in {\gamma}$. Moreover, the spin configurations $f(\Gamma),f(\Gamma'), \sigma_0,\sigma_0'$ are all the same outside ${\gamma}$. This implies that the number of such paths $\mathcal{P}_{\Gamma,\Gamma'}$ is upper bounded by the number of polymers containing $v$.

Now fix some $(\Gamma,\Gamma') \in E^*(\PD)$ such that $(\Gamma_0,\Gamma_0') \in \mathcal{P}_{\Gamma,\Gamma'}$ and assume that $\Gamma' = \Gamma \cup \gam$ for some polymer $\gam\in \Omega$ with $v\in {\gamma}$. Then,
\begin{equation}\label{eq:length-of-path}
|\mathcal{P}_{\Gamma,\Gamma'}| \leq |\gam| \leq M.
\end{equation}
As the path $\mathcal{P}_{\Gamma,\Gamma'}$ is obtained by changing the spins vertex by vertex in the corresponding spin configurations, $f(\Gamma)$ and $f(\Gamma_0)$ differ at most $|\gam|$ vertices. The condition of the theorem implies that
\begin{equation}\label{eq:ratio-of-density}
\frac{\mu_G(\Gamma)}{\mu_G(\Gamma_0)} \leq \eta^{|\gam|} \leq \eta^M.
\end{equation}
The update rule of the restricted Glauber dynamics gives
\begin{equation}\label{eq:rGD-update}
P(\Gamma_0,\Gamma_0') = \frac{1}{n} \cdot \frac{1}{q} \cdot \min \left\{ 1, \frac{w(\Gamma_0')}{w(\Gamma_0)} \right\} \geq \frac{1}{\eta qn}
\end{equation}
and for the polymer dynamics we have
\begin{equation}\label{eq:PD-update}
P'(\Gamma,\Gamma') = \frac{|\gam|}{n} \cdot \frac{1}{2} \cdot w_\gam = \frac{|\gam|w_\gam}{2n}.
\end{equation}

Let $\gam_v$ denote a polymer on $\{v\}$ with a spin from $\{0,\dots,q-1\}\backslash g_v$. Then, the polymer mixing condition implies that $\sum_{\gam\nsim\gam_v} |\gam| w_\gam \leq \theta|\gam_v| < 1$ for some $\theta\in(0,1)$. 
Combining this and inequalities \eqref{eq:length-of-path}, \eqref{eq:ratio-of-density}, \eqref{eq:rGD-update} and \eqref{eq:PD-update}, we get
\[
A(\Gamma_0,\Gamma_0') \leq \sum_{\gam:\, v\in{\gam}} \eta^M \cdot \eta qn \cdot \frac{|\gam|w_\gam}{2n} \cdot M = \frac{1}{2}q M\eta^{M+1} \sum_{\gam\nsim \gam_v} |\gam| w_\gam \leq \frac{1}{2}qM\eta^{M+1}.
\]
For the case where $\sigma_0(v) \neq g_v$ and $\sigma_0'(v) = g_v$, the proof is almost the same and we can get the same bound. Thus,
\[
A = \max\limits_{(\Gamma_0,\Gamma_0') \in E^*(\rGD)} A(\Gamma_0,\Gamma_0') \leq \max\left\{ q, \frac{1}{2}qM\eta^{M+1} \right\} \leq qM\eta^{M+1}.
\]
The theorem then follows from Theorem~\ref{thmPolyMix} and Lemma~\ref{lem:comparison} once we notice that $P(\Gamma,\Gamma) \geq 1/q$ and that $\mu_G(\Gamma) \geq 1/(\eta q)^n$ for all $\Gamma\in \Omega$.
\end{proof}

\subsection{Truncated polymer model}
\label{sec:trunc-polymer}

The bound on the mixing time of the restricted Glauber dynamics in Theorem~\ref{thm:rglaubermix} is exponential in the size of the largest polymer which is in general undesirable. For example, in our applications in Section~\ref{secApply}, $M$ was linear in the number of vertices of the host graph. Here, we show that, under the polymer sampling condition, we can restrict our attention to polymers of size $O(\log n)$ in the sense that the partition function as well as the Gibbs distribution of the truncated polymer model are close to
those of  the original polymer model. 

Let $(\cC(G),w)$ be a polymer model on a graph $G$. For $k>0$, define the truncated polymer model $(\cCtrunc{k}(G),w)$ by
\[
\cCtrunc{k}(G) = \{ \gamma \in \cC(G) : |\gamma| \leq k \}.
\]
Also we let
\[
\Omega_k = \left\{ \Gamma \in \Omega \mid \Gamma \subseteq \cCtrunc{k}(G) \right\} = \{ \Gamma \subseteq \cCtrunc{k}(G) \mid \forall \gamma ,\gamma' \in \Gamma, \gamma \sim \gamma' \}
\]
be the set of allowed configurations (note that $\Omega_k \subseteq \Omega$). 
The partition function of the truncated polymer model $(\cCtrunc{k}(G),w)$ is given by
\[
Z_k(G) = \sum_{\Gamma \in \Omega_k} \prod_{\gamma \in \Gamma} w_\gam.
\]
The corresponding Gibbs distribution on $\Omega_k$ is defined by $\mu_{G,k}(\Gamma) = \frac{\prod_{\gamma \in \Gamma} w_\gamma}{Z_k(G)}$.
We remark that if the original polymer model satisfies the polymer sampling condition then so does the truncated polymer model, and thus Theorem~\ref{thmPolySample} also applies to the truncated model.

The following lemma asserts that the Gibbs distribution and the partition function of the truncated polymer model $(\cCtrunc{k}(G),w)$ are close to those of the original model $(\cC(G),w)$, provided that the polymer sampling condition \eqref{eqSampleCondition} holds.
\begin{lemma}
	\label{lem:hctruncate-polymer}
Let $\mathcal{G}$ be a family of graphs of maximum degree at most $\Delta$
and let $(\cC( \cdot ),w,\mathcal{G})$ be a polymer model that satisfies the 
polymer sampling condition \eqref{eqSampleCondition}
with constant $\tau\geq 5+3\log((q-1)\Delta)$. 	
Let $G$ be an $n$-vertex graph from $\mathcal{G}$.   Then for any $\epsilon > 0$ and $k = \frac{3\log(2n/\eps)}{2\tau}$,  we have
\[Z_k(G) \leq Z(G) \leq e^{\eps} Z_k(G).\]
Moreover, the total variation distance between $\mu_G$ and $\mu_{G,k}$ is at most $\eps$.
\end{lemma}
\begin{proof}
	Note that $Z_k(G) \leq Z(G)$ follows immediately from $\Omega_k \subseteq \Omega$. 
	For $\Gamma \in \Omega_k$, let $\Omega(\Gamma) = \{ \Gamma' \in \Omega : \Gamma' \cap \cCtrunc{k}(G) = \Gamma \}$ and let
	\begin{equation}
	\label{eq:hcpolypfGsplit-polymer}
	\zeta(\Gamma) = \sum_{\Gamma' \in \Omega(\Gamma)} \prod_{\gamma \in \Gamma'} w_\gamma, \mbox{ so that }Z(G) = \sum_{\Gamma \in \Omega_k} \zeta(\Gamma).
	\end{equation}
	Let $\cCtrunc{k}^+(G) = \cC(G) \setminus \cCtrunc{k}(G) = \{ \gamma \in \cC(G) : |\gamma| > k \}$ be the collection of all polymers of size greater than $k$. 
	Notice that for each $\Gamma \in \Omega_k$ we have the crude bound
	\begin{equation}
	\label{eq:splitbound-polymer}
	\zeta(\Gamma) \leq \prod_{\gamma \in \Gamma} w_\gamma \prod_{\gamma \in \cCtrunc{k}^+(G)} (1 + w_\gam).
	\end{equation}	
	Combining~\eqref{eq:hcpolypfGsplit-polymer} and~\eqref{eq:splitbound-polymer}, we obtain that
	\begin{equation}
	\label{eq:partition-trunc-polymer}
		Z(G) \leq \prod_{\gamma \in \cCtrunc{k}^+(G)} (1 + w_\gam) \sum_{\Gamma \in \Omega_k} \prod_{\gamma \in \Gamma} w_\gamma = Z_k(G) \prod_{\gamma \in \cCtrunc{k}^+(G)} (1 + w_\gam).
	\end{equation}
	
	Since $(1 + x) \leq e^x$ for all real $x$, we have that
	\begin{equation}
	\label{eq:logbound-polymer}
	\log \bigg( \prod_{\gamma \in \cCtrunc{k}^+(G)} (1 + w_\gamma) \bigg) \leq \sum_{\gamma \in \cCtrunc{k}^+(G)} w_\gamma \leq \sum_{\gamma \in \cCtrunc{k}^+(G)} e^{-\tau |\gamma|}.
	\end{equation} 	The last inequality follows from the fact that $(\cC( \cdot ), w,\mathcal{G})$ satisfies the polymer sampling condition with constant~$\tau$.
	Then we deduce from Lemma~\ref{lem:BCKL} that
	\begin{equation}
	\label{eq:splitonsize-polymer}
	\sum_{\gamma \in \cCtrunc{k}^+(G)} e^{-\tau |\gamma|} = \sum_{v\in V}\sum_{\ell \geq k} \sum_{\substack{\gamma \in \cCtrunc{k}^+(G) : \\ |\gamma| = \ell, \, v \in  {\gam}}} e^{-\tau \ell} \leq n \,\sum_{\ell \geq k} \left( \frac{e(q-1)\Delta}{e^{\tau}} \right)^\ell,
	\end{equation}
	and since $\tau \geq 5 + 3\log((q-1)\Delta)$, we get $e(q-1)\Delta \leq e^{\tau/3}$. It follows that
	\begin{equation}
	\label{eq:sumbound-polymer}
	\sum_{\ell \geq k} \left( \frac{e(q-1)\Delta}{e^{\tau}} \right)^\ell \leq 2 \exp\left( -\frac{2}{3} \tau k \right) = \frac{\eps}{n}.
	\end{equation}
	Combining~\eqref{eq:partition-trunc-polymer},~\eqref{eq:logbound-polymer},~\eqref{eq:splitonsize-polymer}, and~\eqref{eq:sumbound-polymer} yields $Z(G) \leq e^{\eps} Z_k(G)$, as needed. Finally, we bound the total variation distance between $\mu_{G}$ and $\mu_{G,k}$:
\[\| \mu_{G} - \mu_{G,k} \|_{TV} = \mu_{G}(\Omega \setminus \Omega_k) = \frac{Z(G) - Z_k(G)}{Z(G)} \leq 1 - e^{-\eps} \leq \eps\]
where the first equality is because $\mu_{G}(\Gamma) > \mu_{G,k}(\Gamma)$ if and only if $\Gamma \in \Omega \setminus \Omega_k$, for which we have $\mu_{G,k}(\Gamma) = 0$. This finishes the proof.
\end{proof}

\subsection{Applications}

In this section, we apply the previous results to show that (spin) Glauber dynamics for the ferromagnetic Potts and hard-core models mix in polynomial time on expander graphs, when restricted to configurations close to the ground states (which, as we have already seen, constitute the main portion of the probability space at low temperatures).

\subsubsection{Restricted Glauber for ferromagnetic Potts}\label{sec:f3f344}
In this section, we prove Theorem~\ref{thmPottsGlauber} for the $q$-color ferromagnetic Potts model. Throughout this section, we will work under the assumptions/conditions of Theorem~\ref{thmPotts}. 
That is, we fix a real number $\alpha>0$, integers $q \geq 2$ and $\Delta \geq 3$ and a real number
$\beta\geq \frac{5+3\log((q-1)\Delta)}{\alpha}$.
We let $\mathcal{G}$ be the class of $\alpha$-expander graphs~$G$ with maximum degree at most~$\Delta$.

Let $G$ be an $n$-vertex graph in $\mathcal{G}$
and let $\epsilon$ be a value in $(q e^{-n},1)$.
As in Section~\ref{sec:appPotts}, we will consider the polymer model $(\mathcal{C}^g,w^g)$ whose polymers are connected subgraphs of $G$ with at most $n/2$  vertices, which are labeled by the remaining colors $[q] \setminus \{ g\}$. In fact, following Section~\ref{sec:trunc-polymer}, we will work with a truncation of this model. Namely, for $M>0$, let $(\mathcal{C}^g_M,w^g)$ be the polymer model on $G$ restricted to polymers of size at most $M$.

\begin{observation}
For every $M>0$, the set $\Omega^{g}_M(G)$, as defined in Section~\ref{sec:firstcomparison}, is precisely the set of allowable polymer configurations in the truncated polymer model $(\mathcal{C}^g_M,w^g)$.
\end{observation}

\thmPottsGlauber*
\begin{proof}
We let $\mathcal{G}$ be the class of $\alpha$-expanders with maximum degree at most~$\Delta$.
For the given $n$-vertex graph $G\in \mathcal{G}$,
let $\mu^g_{G,M}$ be the Gibbs distribution of the polymer model $(\mathcal{C}^g_M(G),w^g)$.

By Lemma~\ref{SampleConditionPotts}, we have that $(\mathcal{C}^g(\cdot),w^g,\mathcal{G})$ satisfies the polymer sampling condition with $\tau=\alpha \beta$ and hence so does the truncated polymer model $(\mathcal{C}^g_M(\cdot),w^g,\mathcal{G})$. The result therefore follows by applying Theorem~\ref{thm:rglaubermix}, after observing that (i) the polymer model $(\mathcal{C}^g_M(G),w^g)$ is single-update-compatible (use DFS ordering), and (ii) for a pair of polymer configurations $\Gamma,\Gamma'\in \Omega^{g}_M$  whose corresponding spin configurations $\sigma,\sigma'$ differ at a vertex, we have
\begin{equation*}
\frac{\mu^g_{G,M}(\Gamma)}{\mu^{g}_{G,M}(\Gamma')}=\frac{e^{-\beta m(G,\sigma) }}{e^{-\beta m(G,\sigma') }}\in [1/\eta,\eta] 
\end{equation*}
where $\eta=\exp(\beta \Delta)$ (since $G$ has maximum degree $\Delta$, changing the spin of a vertex can create at most $\Delta$ new monochromatic/bichromatic edges). This finishes the proof.
\end{proof}

The following lemma justifies that the set $\Omega^{g}_M(G)$ with $M=O(\log (n/\epsilon))$ constitutes for all but $\epsilon$ of the aggregate weight of colorings in the Potts distribution on $G$. 

\begin{lemma}\label{lem:f34f34}
Suppose $q\geq 2$, $\Delta \geq 3$ are integers and $\alpha>0$ is a real. Let $\beta \geq \frac{5+3\log((q-1) \Delta)}{\alpha}$ be a real number and $g \in [q]$. Then, for any $n$-vertex $\alpha$-expander graph $G$ of maximum degree $\Delta$ and any $q e^{-n}\leq \epsilon<1$, for $M=\frac{3\log(4n/\eps)}{2\alpha\beta}$,  $q Z^g_M(G)$ is an $\epsilon$-approximation of the Potts partition function $Z_{G,\beta}$.
\end{lemma}
\begin{proof}
Let $Z^g(G)$ be the partition function of the polymer model $(\mathcal{C}^g(G),w^g)$. Since $\beta\geq\frac{ 5+ 3 \log ((q-1)\Delta)   }{\alpha   } > \frac{2 \log (eq)}{\alpha}$ and $\epsilon\geq q e^{-n}$, by Lemma~\ref{lemPottsApprox} we have that $qZ^g(G)$ is an $(\epsilon/2)$-approximation to $Z_{G,\beta}$. 
If $\mathcal{G}$ is the class of $\alpha$-expanders with maximum degree at most~$\Delta$ then
by Lemma~\ref{SampleConditionPotts}, we have that $(\mathcal{C}^g(\cdot),w^g,\mathcal{G})$ satisfies the polymer sampling condition with $\tau=\alpha \beta$ and hence so does the truncated polymer model $(\mathcal{C}^g_M(\cdot),w^g,\mathcal{G})$. It follows by Lemma~\ref{lem:hctruncate-polymer} that, for $M=\frac{3\log(4n/\eps)}{2\alpha\beta}$,  $q Z^g_M(G)$ is an  $(\epsilon/2)$-approximation to $qZ^g(G)$. Therefore, $q Z^g_M(G)$ is an $\epsilon$-approximation to $Z_{G,\beta}$.
\end{proof}

\subsubsection{Restricted Glauber dynamics for hard-core mixes in polynomial time}
In this section, we state and prove the analogue of Theorem~\ref{thmPottsGlauber} for the hard-core model. In particular, let $G=(\Vpart{0},\Vpart{1},E)$ be an $n$-vertex $\alpha$-expanding bipartite graph of maximum degree $\Delta$, and for $i\in \{0,1\}$ and $M>0$, let $\Omega^{i}_M(G)$ denote the independent sets $I$ whose deviations from the ground state $\Vpart{i}$ consists of small connected components, more precisely, $(\Vpart{i}\backslash I) \cup (I\cap \Vpart{1-i})$ consists of connected components of size at most $M$. Using similar methods to Section~\ref{sec:f3f344}, we will show the following.
\begin{theorem}\label{thmHCGlauber}
Fix $\alpha>0$, $\Delta \geq 3$ and $\lambda\geq (6\Delta)^{3+6/\alpha}$.
For any $n$-vertex bipartite $\alpha$-expander with maximum degree at most~$\Delta$ and any
$\epsilon\in (0,1)$ and $i\in\{0,1\}$, with $M = O(\log(n/\eps))$, the Glauber dynamics restricted to $\Omega^i_M(G)$ has mixing time $\Tmix(\eps)$ polynomial in $n$ and $1/\eps$.
\end{theorem}
As we shall see soon in the upcoming Lemma~\ref{lem:3tt5b4666}, and for $\lambda$ large enough, the set $\Omega^0_M(G)\cup\Omega^1_M(G)$ for $M=\Theta(\log(n/\eps))$ captures all but $\epsilon$ weight of the hard-core partition function and hence Theorem~\ref{thmHCGlauber} can be used to obtain another polynomial time algorithm for the hard-core model on expanding graphs $G$ in that regime.

To prove Theorem~\ref{thmHCGlauber}, it will be simpler to work with somewhat different polymer models than those in Section~\ref{sec:apphardcore}. These models were originally used in~\cite{JKP} (the conference version of \cite{JKP2}). For $i\in \{0,1\}$, and following \cite{JKP}, we will define a polymer model $(\Cpart{i}(G),\wpart{i})$.  The host graph will be $G$ and the model will capture deviations from the ground state $\Vpart{i}$: a polymer $\gamma$ will  be a connected set of vertices in $G$ such that $(\Vpart{i}\backslash I) \cup (I\cap \Vpart{1-i})=\gamma$ for some independent set $I$. Specifically, the set $\Cpart{i} = \Cpart{i}(G)$ of allowed polymers consists of all connected sets of vertices $\gamma$ (in $G$) such that $| \gamma \cap \Vpart{i}| \leq |\Vpart{i}|/4$ and for any $v\in \gamma\cap \Vpart{\opposite{i}}$, all of the neighbors of $v$ (in $G$) are also in $\gamma$. The set of spins is $\{0,1\}$ and the ground state spin for a vertex $v$ is $g_v=1$ if $v\in \Vpart{i}$, and $g_v=0$ if $v\in \Vpart{1-i}$; the spin assignment $\sigma_\gamma$ for a polymer $\gamma$ is given by $1- g_v$ for $v\in \gamma$.   The weight of a polymer $\gamma\in \Cpart{i}$ is defined as 
\[
\wpart{i}_\gamma = \frac{\lambda^{| \gamma \cap \Vpart{\opposite{i}}|}}{\lambda^{|\gamma \cap \Vpart{i}|}}.
 \] 
The main observation behind the definition of the weight $\wpart{i}_\gamma$ is that the weight of an independent set $I$ such that $(\Vpart{i}\backslash I) \cup (I\cap \Vpart{1-i})=\gamma$ is $\lambda^{|\Vpart{i}|}\wpart{i}_\gamma$. 

Following again Section~\ref{sec:trunc-polymer}, it will be relevant to consider, for $M>0$,  the truncated polymer model $(\mathcal{C}^i_M(G),w^i)$ whose polymers are of size at most $M$; observe that the set $\Omega^i_M$ defined above is precisely the set of allowable polymer configurations in the truncated polymer model. We next verify the polymer sampling condition \eqref{eqSampleCondition} for these models and conclude the proof of Theorem~\ref{thmHCGlauber}.
\begin{lemma}\label{lem:4t4g4}
Fix $\alpha>0$ and $\Delta \geq 3$. Let $\mathcal{G}$ be the class of bipartite $\alpha$-expanders with maximum degree at most~$\Delta$.
For $\lambda\geq (6\Delta)^{3+6/\alpha}$ and $i\in \{0,1\}$, the polymer model $(\mathcal{C}^i(\cdot), w^i,\mathcal{G})$ satisfies the polymer sampling condition \eqref{eqSampleCondition} with $\tau=\frac{\alpha}{2+\alpha} \log \lambda$.
\end{lemma}
\begin{proof}
We have $\tau\geq 5+3\log \Delta$, so it suffices to show that for $G\in \mathcal{G}$ and $\gamma\in(\mathcal{C}^i(G), w^i)$ it holds that $w^i_\gamma\leq e^{-\tau|\gamma|}$.  For $v\in \gamma\cap \Vpart{\opposite{i}}$ we have that all of the neighbors of $v$  are also in $\gamma$ and hence, by the $\alpha$-expansion of $G$, we have that  $|\gamma\cap \Vpart{i}|\geq (1+\alpha)|\gamma\cap \Vpart{1-i}|$. This gives $(2+\alpha)|\gamma\cap \Vpart{i}|\geq (1+\alpha)|\gamma|$ and therefore
\[w^i_\gamma=\frac{\lambda^{| \gamma \cap \Vpart{\opposite{i}}|}}{\lambda^{|\gamma \cap \Vpart{i}|}}\leq \frac{\lambda^{\frac{1}{1+\alpha}|\gamma \cap \Vpart{i}|}}{\lambda^{|\gamma \cap \Vpart{i}|}}\leq \lambda^{-\frac{\alpha}{2+\alpha}|\gamma|}=e^{-\tau|\gamma|}.\qedhere\]
\end{proof}

\begin{proof}[Proof of Theorem~\ref{thmHCGlauber}]
Let $\mathcal{G}$ be the class of bipartite $\alpha$-expanders with maximum degree at most~$\Delta$.
Consider an $n$-vertex graph $G\in \mathcal{G}$ and let
  $\mu^i_{G,M}$ be the Gibbs distribution of the model $(\mathcal{C}^i_M(G),w^i)$.

By Lemma~\ref{lem:4t4g4}, we have that $(\mathcal{C}^i(\cdot),w^i,\mathcal{G})$ satisfies the polymer sampling condition with $\tau=\frac{\alpha}{2+\alpha} \log \lambda$ and hence so does the truncated polymer model $(\mathcal{C}^i_M(\cdot),w^i,\mathcal{G})$.  The result therefore follows by applying Theorem~\ref{thm:rglaubermix}, after observing that (i) the polymer model $(\mathcal{C}^i_M(G),w^i)$ is single-update-compatible (use DFS ordering), and (ii) for a pair of polymer configurations $\Gamma,\Gamma'\in \Omega^{i}_M$  whose corresponding independent sets  $I, I'$ differ in at most one vertex, we have
\begin{equation*}
\frac{\mu^i_{G,M}(\Gamma)}{\mu^{i}_{G,M}(\Gamma')}=\frac{\lambda^{|I|}}{\lambda^{|I'|}}\in [1/\lambda,\lambda].\qedhere
\end{equation*}
\end{proof}
Finally, we justify that, for large enough $\lambda$ and $M=\Theta(\log(n/\eps))$, the aggregate weight of independent sets in $\Omega^0_M(G)\cup\Omega^1_M(G)$ captures all but $\epsilon$ fraction of the hard-core partition function $Z_{G,\lambda}$.  Let $Z^i(G)$ denote the partition function of the polymer model $(\Cpart{i}, \wpart{i})$ and $Z^i_M(G)$ denote the partition function of the polymer model $(\Cpart{i}_M, \wpart{i}_M)$. 
We will need the following lemma from \cite{JKP}.
\begin{lemma}[{\cite[Lemmas 4.1 \& 4.2]{JKP}}]
\label{lemHCApprox}
For $\lambda>\max\big\{(2e)^{\frac{8n}{\alpha n_0}},(2e)^{\frac{8n}{\alpha n_1}}, (2e)^{(40/\alpha)}\big\}$, the number  $\lam^{n_0} Z^0(G) + \lam^{n_1} Z^1(G)$ is a $(2e^{-n})$-approximation of the hard-core partition function $Z_{G,\lambda}$, where $n_i = |\Vpart{i}|$ for $i\in \{0,1\}$.
\end{lemma}
\begin{lemma}\label{lem:3tt5b4666}
Fix $\alpha\in (0,1)$ and $\Delta \geq 3$.
There exists a constant $C > 0 $ such that for $\lambda > 
 (6C\Delta)^{3 + 6/\alpha} $ the following holds for all $n$-vertex bipartite $\alpha$-expander graphs $G=(V^0,V^1,E)$ of maximum degree at most $\Delta$. 
For all $\epsilon\in (4e^{-n},1)$ and $M=\frac{3(2+\alpha)\log(4n/\eps)}{2\alpha\log \lambda}$,  the number  \[\hat{Z}=\lam^{n_0} Z^0_M(G) + \lam^{n_1} Z^1_M(G)\] is an $\epsilon$-approximation of the hard-core partition function $Z_{G,\lambda}$, where $n_i = |\Vpart{i}|$ for $i\in\{0,1\}$.
\end{lemma}
\begin{proof}
Let $n_i = |\Vpart{i}|$ for $i\in \{0,1\}$ and observe that $n/n_0, n/n_1\leq 3$ (using that $G$ is an $\alpha$-expander for $\alpha\in (0,1)$, see \cite{JKP} for details). Therefore, by taking $C$ large enough, we have  that for all $\lambda > 
 (6C\Delta)^{3 + 6/\alpha} $ both Lemmas~\ref{lem:4t4g4} and~\ref{lemHCApprox} apply. Let $\epsilon\in (4e^{-n},1)$. 

By Lemma~\ref{lemHCApprox}, we have that $\lam^{n_0} Z^0(G) + \lam^{n_1} Z^1(G)$ is an $(\epsilon/2)$-approximation to $Z_{G,\lambda}$. By Lemma~\ref{lem:4t4g4}, we have that, for $i\in\{0,1\}$, $(\mathcal{C}^i(\cdot),w^i,\mathcal{G})$ satisfies the polymer sampling condition with $\tau=\frac{\alpha}{2+\alpha}\log \lambda$ and hence so does the truncated polymer model $(\mathcal{C}^i_M(\cdot),w^i,\mathcal{G})$.
(Here, as usual, we take
$\mathcal{G}$ to be the class of bipartite $\alpha$-expanders with maximum degree at most~$\Delta$.)
 It follows by Lemma~\ref{lem:hctruncate-polymer} that, for $M=\frac{3(2+\alpha)\log(4n/\eps)}{2\alpha\log \lambda}$,  $ Z^i_M(G)$ is an  $(\epsilon/2)$-approximation to $Z^i(G)$. Therefore, $\hat{Z}$ is an $\epsilon$-approximation to $Z_{G,\lambda}$.
\end{proof}

\bibliographystyle{plain}
\bibliography{\jobname}

\begin{thebibliography}{10}

\bibitem{Barvinokbook}
A.~Barvinok.
\newblock {\em Combinatorics and Complexity of Partition Functions}.
\newblock Algorithms and Combinatorics. Springer International Publishing,
  2017.

\bibitem{BSVV2008}
I.~Bez{\'a}kov{\'a}, D.~{\v{S}}tefankovi{\v{c}}, V.~V. Vazirani, and E.~Vigoda.
\newblock Accelerating simulated annealing for the permanent and combinatorial
  counting problems.
\newblock {\em SIAM Journal on Computing}, 37(5):1429--1454, 2008.

\bibitem{bordewich2016mixing}
M.~Bordewich, C.~Greenhill, and V.~Patel.
\newblock Mixing of the {G}lauber dynamics for the ferromagnetic {P}otts model.
\newblock {\em Random Structures \& Algorithms}, 48(1):21--52, 2016.

\bibitem{borgs2006absence}
C.~Borgs.
\newblock Absence of zeros for the chromatic polynomial on bounded degree
  graphs.
\newblock {\em Combinatorics, Probability and Computing}, 15(1-2):63--74, 2006.

\bibitem{borgs1999torpid}
C.~Borgs, J.~T. Chayes, A.~Frieze, J.~H. Kim, P.~Tetali, E.~Vigoda, and V.~H.
  Vu.
\newblock Torpid mixing of some {M}onte {C}arlo {M}arkov chain algorithms in
  statistical physics.
\newblock In {\em Proceedings of the 40th Annual IEEE Symposium on Foundations
  of Computer Science (FOCS)}, pages 218--229, 1999.

\bibitem{BCKL}
C.~Borgs, J.~T. Chayes, J.~Kahn, and L.~Lov{\'a}sz.
\newblock Left and right convergence of graphs with bounded degree.
\newblock {\em Random Structures \& Algorithms}, 42(1):1--28, 2013.

\bibitem{borgs1989unified}
C.~Borgs and J.~Z. Imbrie.
\newblock A unified approach to phase diagrams in field theory and statistical
  mechanics.
\newblock {\em Communications in mathematical physics}, 123(2):305--328, 1989.

\bibitem{DS2}
P.~Diaconis and L.~Saloff-Coste.
\newblock Comparison techniques for random walk on finite groups.
\newblock {\em Ann. Probab.}, 21(4):2131--2156, 1993.

\bibitem{DS1}
P.~Diaconis and L.~Saloff-Coste.
\newblock Comparison theorems for reversible {M}arkov chains.
\newblock {\em Ann. Appl. Probab.}, 3(3):696--730, 1993.

\bibitem{dobrushin1996estimates}
R.~L. Dobrushin.
\newblock Estimates of semi-invariants for the {I}sing model at low
  temperatures.
\newblock {\em Translations of the American Mathematical Society-Series 2},
  177:59--82, 1996.

\bibitem{comparisonsurvey}
M.~E. Dyer, L.~A. Goldberg, M.~Jerrum, and R.~Martin.
\newblock Markov chain comparison.
\newblock {\em Probab. Surveys}, 3:89--111, 2006.

\bibitem{DG}
M.~E. Dyer and C.~S. Greenhill.
\newblock {\em Random Walks on Combinatorial Objects}, pages 101--136.
\newblock London Mathematical Society Lecture Note Series. Cambridge University
  Press.

\bibitem{DGIndSet}
M.~E. Dyer and C.~S. Greenhill.
\newblock On {M}arkov chains for independent sets.
\newblock {\em J. Algorithms}, 35(1):17--49, 2000.

\bibitem{EHSVY}
C.~Efthymiou, T.~P. Hayes, D.~{\v{S}}tefankovi{\v{c}}, E.~Vigoda, and Y.~Yin.
\newblock Convergence of {MCMC} and loopy {BP} in the tree uniqueness region
  for the hard-core model.
\newblock In {\em Proceedings of the 57th Annual IEEE Symposium on Foundations
  of Computer Science (FOCS)}, pages 704--713, 2016.

\bibitem{fernandez2001loss}
R.~Fern{\'a}ndez, P.~A. Ferrari, and N.~L. Garcia.
\newblock Loss network representation of {P}eierls contours.
\newblock {\em Annals of Probability}, pages 902--937, 2001.

\bibitem{GT}
D.~Galvin and P.~Tetali.
\newblock Slow mixing of {G}lauber dynamics for the hard-core model on regular
  bipartite graphs.
\newblock {\em Random Structures \& Algorithms}, 28(4):427--443, 2006.

\bibitem{gruber1971general}
C.~Gruber and H.~Kunz.
\newblock General properties of polymer systems.
\newblock {\em Communications in Mathematical Physics}, 22(2):133--161, 1971.

\bibitem{HPR}
T.~Helmuth, W.~Perkins, and G.~Regts.
\newblock Algorithmic {P}irogov-{S}inai theory.
\newblock In {\em Proceedings of the 51st Annual ACM SIGACT Symposium on Theory
  of Computing}, STOC 2019, pages 1009--1020, 2019.

\bibitem{H2015}
M.~Huber.
\newblock Approximation algorithms for the normalizing constant of {G}ibbs
  distributions.
\newblock {\em The Annals of Applied Probability}, 25(2):974--985, 2015.

\bibitem{JKP}
M.~Jenssen, P.~Keevash, and W.~Perkins.
\newblock Algorithms for \#{BIS}-hard problems on expander graphs.
\newblock In {\em Proceedings of the Thirtieth Annual ACM-SIAM Symposium on
  Discrete Algorithms (SODA)}, pages 2235--2247, 2019.

\bibitem{JKP2}
M.~Jenssen, P.~Keevash, and W.~Perkins.
\newblock Algorithms for \#{BIS}-hard problems on expander graphs.
\newblock {\em arXiv preprint}, 1807.04804v2, 2019.

\bibitem{KP}
R.~Koteck\'{y} and D.~Preiss.
\newblock Cluster expansion for abstract polymer models.
\newblock {\em Comm. Math. Phys.}, 103(3):491--498, 1986.

\bibitem{laanait1991interfaces}
L.~Laanait, A.~Messager, S.~Miracle-Sol{\'e}, J.~Ruiz, and S.~Shlosman.
\newblock Interfaces in the {P}otts model {I}: {P}irogov-{S}inai theory of the
  {F}ortuin-{K}asteleyn representation.
\newblock {\em Communications in Mathematical Physics}, 140(1):81--91, 1991.

\bibitem{liao2019counting}
C.~Liao, J.~Lin, P.~Lu, and Z.~Mao.
\newblock Counting independent sets and colorings on random regular bipartite
  graphs.
\newblock {\em arXiv preprint arXiv:1903.07531}, 2019.

\bibitem{LiuLu}
J.~Liu and P.~Lu.
\newblock {FPTAS} for {\#}{BIS} with degree bounds on one side.
\newblock In {\em Proceedings of the 47th Annual {ACM} on Symposium on Theory
  of Computing (STOC)}, pages 549--556, 2015.

\bibitem{mossel2009hardness}
E.~Mossel, D.~Weitz, and N.~Wormald.
\newblock On the hardness of sampling independent sets beyond the tree
  threshold.
\newblock {\em Probability Theory and Related Fields}, 143(3-4):401--439, 2009.

\bibitem{PR}
V.~Patel and G.~Regts.
\newblock Deterministic polynomial-time approximation algorithms for functions
  and graph polynomials.
\newblock {\em SIAM Journal on Computing}, 46(6):1893--1919, 2017.

\bibitem{PS}
S.~A. Pirogov and Ya.~G. Sinai.
\newblock Phase diagrams of classical lattice systems.
\newblock {\em Teoret. Mat. Fiz.}, 25(3):358--369, 1975.

\bibitem{RT}
D.~Randall and P.~Tetali.
\newblock Analyzing {G}lauber dynamics by comparison of {M}arkov chains.
\newblock {\em J. Math. Phys.}, 41(3):1598--1615, 2000.
\newblock Probabilistic techniques in equilibrium and nonequilibrium
  statistical physics.

\bibitem{SVV2009annealing}
D.~{\v{S}}tefankovi{\v{c}}, S.~Vempala, and E.~Vigoda.
\newblock Adaptive simulated annealing: A near-optimal connection between
  sampling and counting.
\newblock {\em Journal of the ACM}, 56(3):18, 2009.

\bibitem{Weitz}
D.~Weitz.
\newblock Counting independent sets up to the tree threshold.
\newblock In {\em Proceedings of the 38th Annual {ACM} Symposium on Theory of
  Computing (STOC)}, pages 140--149, 2006.

\end{thebibliography}

\end{document}